 \documentclass[letterpaper, 10 pt, conference]{IEEEtran}  
 \usepackage[utf8]{inputenc}
\pdfoutput=1 
\IEEEoverridecommandlockouts                              
\usepackage{tikz}
\usepackage{cite}
\usepackage{amsmath,amssymb,amsfonts,mathtools}
\usepackage{dsfont}
\usepackage[normalem]{ulem}
\usepackage{algorithmic}
\usepackage{graphicx}
\usepackage{textcomp}
\usepackage{xcolor}
\usepackage{array}
\usepackage[english]{babel}
\graphicspath{ {./} }
\usepackage{centernot}
\usepackage{caption}
\usepackage{subcaption}
\usepackage{bm}
\usepackage{adjustbox,lipsum}
\usepackage{flushend}
\usepackage{multirow}

\def\BibTeX{{\rm B\kern-.05em{\sc i\kern-.025em b}\kern-.08em
    T\kern-.1667em\lower.7ex\hbox{E}\kern-.125emX}}

\usepackage{theorem}
\usepackage{comment}
\newtheorem{theorem}{Theorem}
\newtheorem{definition}{Definition}
\newtheorem{lemma}{Lemma}
\newtheorem{corollary}{Corollary}
\newtheorem{proposition}{Proposition}
\newtheorem{assumption}{Assumption}
{\theorembodyfont{\upshape}
	\newtheorem{remark}{Remark}
	\newtheorem{example}{Example}
	
}
\usepackage{stackengine}
\stackMath
\def\qed{\hfill \vrule height 5pt width 5pt depth 0pt \medskip}
\usepackage{color} 
\usepackage{soul} 

\newcommand{\bite}{\begin{itemize}}
\newcommand{\eite}{\end{itemize}}
\newcommand{\benu}{\begin{enumerate}}
\newcommand{\eenu}{\end{enumerate}}

 \makeatletter
\let\saved@bibitem\@bibitem
\makeatother

\usepackage{bibentry}
\usepackage[colorlinks=true,urlcolor=blue,linkcolor=black,citecolor=black]{hyperref}
 
\definecolor{darkpastelgreen}{rgb}{0.01, 0.75, 0.24}

\newcommand{\proof}{\noindent {\bf Proof. }}
\title{\LARGE \bf
Wisdom of Crowds Effects under Antagonistic Interactions and Correlated Opinions}

\author{Muhammad Ahsan Razaq$^{1}$ and Claudio Altafini$^{1}$
\thanks{Work supported in part by grants from the Swedish Research Council (grants n. 2020-03701 and 2024-04772 to C.A.).}
\thanks{$^{1}$ The authors are with the Division of Automatic Control, Department of Electrical Engineering, Linköping University, SE-58183 Linköping, Sweden
        {\tt\small(email: muhammad.ahsan.razaq@liu.se; claudio.altafini@liu.se)}}%
}

\newcommand\numberthis{\addtocounter{equation}{1}\tag{\theequation}}

\begin{document}

\maketitle
\thispagestyle{empty}
\pagestyle{empty}

\begin{abstract}
This paper investigates the wisdom of crowds of linear opinion dynamics models evolving on signed networks. Conditions are given under which models such as the DeGroot, Friedkin-Johnsen (FJ) and concatenated FJ models improve or undermine collective wisdom. The extension to dependent initial opinions is also presented, highlighting how the correlation structure influences the feasibility and geometry of the wisdom-improving regions.
\end{abstract}
\section{Introduction}

The wisdom of crowds effect describes the situation in which a group’s collective judgment can surpass the accuracy and reliability of all individual assessments made in isolation \cite{surowiecki2005wisdom}. 
This effect is achieved through an aggregation mechanism, typically some form of averaging, which combines together the opinions of the individuals, normally assumed to be independent. 
Averaging can occur ``statically'', when a central coordinator can collect the opinions of all users, or in a distributed fashion, through an iterative process.
Opinion dynamics models are natural candidates for the latter task, because they often implement some form of averaging, and are inherently decentralized---individuals update their opinions based solely on their neighbors.
Among them, the simplest choices are linear opinion dynamics models, which provide an elementary, yet principled, framework for studying how individual beliefs evolve through social interactions \cite{Noorazar2020opinion,proskurnikov2017tutorial}. In this paper, we focus on three such models: the DeGroot model \cite{degroot1974reaching}, the Friedkin–Johnsen (FJ) model \cite{friedkin1990social}, and the concatenated FJ model \cite{wang2022consensus}. The DeGroot model describes opinion updates as weighted averaging over a social network, leading to a consensus. The FJ model introduces stubbornness, allowing agents to retain partial attachment to their initial opinions, which prevents consensus but nevertheless leads to a contraction in opinion space.  The concatenated FJ model generalizes the FJ model to repeated, concatenated discussion rounds, enabling opinions to gradually approach consensus over multiple discussions.

Studying the wisdom of crowd effect on an opinion dynamics model amounts to checking when the variance concentrates (i.e., decreases) around the mean value in passing through the dynamics.

Recent work has shown that the social influence exerted by an opinion dynamics model can both enhance and undermine the wisdom of crowds effect \cite{mannes2014wisdom,madirolas2015improving,almaatouq2020adaptive,lorenz2011social,becker2017network,golub2010naive,tian2023dynamics,tian2023social}. Network structure plays a key role: while some configurations with reduced diversity can undermine accuracy \cite{lorenz2011social}, others enable improvement even under constraints \cite{becker2017network}. Strategies such as weighting accurate individuals more heavily \cite{mannes2014wisdom} or increasing their resistance to change \cite{madirolas2015improving} have also been shown to boost collective performance. 
In an opinion dynamics model, these strategies modify the so-called social power of the individuals, a quantity related to eigenvector centrality \cite{friedkin1997social} that describes how strongly each individual influences the final outcome of the averaging process. For instance, the enhancing of collective wisdom through the adaptive adjustments of social power via feedback is studied in \cite{almaatouq2020adaptive}. Tian {\em et al.} \cite{tian2023social} characterize the feasible region of social power allocations that guarantee variance reduction (i.e., improvement in the wisdom), while \cite{tian2023dynamics} establish conditions under which DeGroot, FJ, and concatenated FJ models remain asymptotically ``wise,'' meaning that the group estimates converge to the true value as the population grows. Additionally, Golub {\em et al.} \cite{golub2010naive} identify the structural properties of interaction graphs that make the DeGroot model ``wise''.

So far, the wisdom of crowds effect has been analyzed only in cooperative networks, where all interaction weights are non-negative. However, real-world social systems often include antagonistic relationships, motivating the study of signed social networks.  Recent work has extended classical opinion dynamics models---such as DeGroot \cite{altafini2012consensus,Fontan2022}, FJ \cite{razaq2023propagation,shrinate2024absolute}, and concatenated FJ \cite{razaq2025signed}---to antagonistic networks.  
These papers typically capture negative influence through two distinct mechanisms: opposing and repelling dynamics \cite{shi2019dynamics}. Under opposing dynamics, the system’s behavior is strongly tied to the structural balance of the underlying graph \cite{altafini2012consensus}. In structurally balanced networks, the DeGroot model converges to a bipartite consensus, whereas in unbalanced networks all opinions collapse to zero in an asymptotically stable manner.

In contrast, for repelling dynamics, which is the model of choice for this paper, structural balance does not play any role.  Instead, what decides stability and asymptotic behavior is a Perron–Frobenius (PF) property, which requires that the largest eigenvalue of the signed interaction matrix is real, positive and strictly dominant over all other eigenvalues. 
The associated right eigenvector characterizes the endpoint of the dynamics, while the corresponding left eigenvector represents the social power of the individuals \cite{tian2023dynamics}. 
Both can have negative entries (and not in the same positions). Negative components in the social power vector indicate that the group perceives the influence of the corresponding agents as harmful and assigns negative weights to their contributions in the averaging process. This interpretation extends, with minor adjustments, to affine models such as the FJ model. 
The possibility of negative components in the right PF eigenvector leads to two distinct cases:
\benu
\item {\em Signed unipartite case}: The dominant right eigenvector of the signed interaction matrix is positive (i.e., the aforementioned PF property is ``unsigned''), which in the DeGroot model leads to consensus despite the presence of negative weights;
\item {\em Signed bipartite case}: The dominant right eigenvector has components of mixed signs (we call this signed PF property), meaning that in a DeGroot model we achieve bipartite consensus.
\eenu

For what concerns the wisdom of crowd effect, while the signed unipartite case behaves essentially like the unsigned counterpart, namely we have accuracy in mean and variance concentration (i.e., improvement of the wisdom) in a convex subregion of a hyperellipsoidal function of the variance of the initial conditions, the signed bipartite case instead differs from it, in the sense that accuracy in mean is missing, but in spite of this, the variance can nevertheless concentrate.
This means that by exchanging opinions in a ``signed bipartite graph'', the community will typically fail to reach the truth, but both groups will be nevertheless ``more convinced'' of their own value. 
This is a situation that happens in practice in polarized social communities.

In the paper, we further extend the analysis to situations where the initial opinions are dependent rather than independent, introducing covariances in the initial conditions. This generalization is only possible in the framework of signed graphs, as correlations may naturally lead to  social power vectors that are signed and hence must originate from  signed interaction graphs.
Passing from independent to correlated initial opinions significantly alters the geometry of the concentration region, as correlations reshape the conditions for improving or undermining wisdom. Our results show that, even under strong dependencies, wisdom of crowd can occur, and an optimal allocation of social power can still reduce the group variance, although at the cost of introducing negative weights in the interaction matrix. These findings provide a more realistic framework for understanding opinion aggregation in interconnected societies where independence of beliefs rarely holds.

A preliminary version of this paper will appear in the proceeding of CDC'25 \cite{razaq2025wisdom}. 
This journal version significantly expands upon the conference paper by introducing two major extensions: (i) the analysis of the signed bipartite case, which captures polarization effects and their impact on variance concentration, and (ii) the case of dependent initial opinions, which modifies the geometry of wisdom-improving regions and highlights scenarios where negative social power may become necessary for improving wisdom. 

The rest of the paper is organized as follows. Section \ref{prelimarymaterial} introduces the necessary preliminaries on signed graphs and the PF property. Section \ref{Problemformulation} formalizes the notions of improved and undermined wisdom of crowds. Section \ref{section:wisdomunsignedmodels} reviews the  wisdom of crowds problem for unsigned networks. Building on this, Sections \ref{section:wisdomsignedmodels} and \ref{section:wisdombipartitemodels} extend the analysis to signed networks, focusing on the unipartite and bipartite cases, respectively. Section \ref{Section:dependent} generalizes the framework to account for dependent initial opinions and Section \ref{section:geometricinterpretation} provides a geometric interpretation of the concentration regions. Finally, Section \ref{section:continuous_time} discusses continuous-time counterparts of the proposed models.

\section{Preliminary Material}
\label{prelimarymaterial}
\textit{Notations.} Scalars are denoted by lowercase letters, vectors by boldface lowercase letters, and matrices by capital letters.  The all-ones and all-zeros vector are represented by $\mathds{1}$ and \(\mathbf{0}\), respectively. For a vector \(\mathbf{x}\in\mathbb{R}^n\),  ${\rm diag}\left( \mathbf{x} \right) \in \mathbb{R}^{n \times n}$ denotes a diagonal matrix whose diagonal entries are given by \(\mathbf{x}\). \(H_{ij}\) denotes the \(ij\)th term of the matrix \(H\in \mathbb{R}^{n \times n}\).
The spectrum of the matrix \(H\in\mathbb{R}^{n\times n}\) is denoted by \(\Lambda\left(H\right)=\{\lambda_i\left(H\right)\},\) for \( i=1,\dots,n\).  The spectral radius of \(H\in\mathbb{R}^{n\times n}\), written as \(\rho(H)\), is the largest absolute eigenvalue: $\rho\left(H\right)=\max_{i=1,\dots,n}\left\lvert\lambda_i\left(H\right)\right\rvert$. If $\lvert\lambda_i(H)\rvert>\lvert\lambda_j(H)\rvert$ $\forall \,j\ne i$, then \(\lambda_i(H)\) is called a strictly dominant eigenvalue, and its associated left and right eigenvectors are referred to as the dominant left and right eigenvectors, respectively. 
The Hadamard product (element-wise product) of two matrices \( H_1 \in\mathbb{R}^{n\times n}\) and \( H_2 \in\mathbb{R}^{n\times n} \)  is denoted by \( H_1 \odot H_2 \).
A matrix \( H \in \mathbb{R}^{n \times n} \) is said to be positive (semi-)definite, denoted \( H \succ 0 \) (\( H \succeq 0 \)), if \( \mathbf{x}^\top H \mathbf{x} >  0 \) (\( \mathbf{x}^\top H \mathbf{x} \ge  0 \)) for all nonzero vectors \( \mathbf{x} \in \mathbb{R}^n \). Similarly, it is negative \mbox{(semi-)definite}, denoted \( H \prec 0 \) ( \( H \preceq 0 \)), if \(- H \succ 0 \) (\( -H \succeq 0 \)).
For a compact set $H \subset \mathbb{R}^n$, $\partial H$ and $\mathrm{int}(H)$ denote the boundary and interior of $H$, respectively.

\subsection{Signed Graphs}
A directed graph (digraph) is represented as a triplet \mbox{$\mathcal{G}=\left(\mathcal{V},\mathcal{E},A\right)$,} where $\mathcal{V}=\left\{v_1,v_2,\dots,v_n\right\}$ is a set of $n$ nodes, $\mathcal{E}$ is a set of  directed edges, and \(A\) is the adjacency matrix. An edge $\left(i,j\right) \in \mathcal{E}$ indicates a directed link from node $v_i$ to node $v_j$, and $A_{ij}=0$ whenever $\left(j,i\right) \notin \mathcal{E}$.
In signed graphs, the values of $A_{ij}$ can be positive or negative, representing friendly or antagonistic relationships, respectively.

For a signed graph, the weighted in-degree vector is defined as $\delta_{\rm in}=\left[\sum_{j=1, i \ne j}^n A_{ij}\right]_{n \times 1} \in \mathbb{R}^n$. The ``repelling Laplacian'' \cite{shi2019dynamics} is given by $L=D-A$, where $D={\rm diag}\left(\delta_{\rm in}\right)$. By construction $L\mathds{1}=\mathbf{0}$.

\subsection{Perron-Frobenious property for signed matrices}


The following definitions can be found in \cite{Fontan2022,razaq2023propagation}.
\begin{itemize}
    \item The matrix $H \in \mathbb{R}^{n \times n}$ satisfies the right (resp. left) Perron-Frobenius (PF) property (denoted $H \in \mathcal{PF}$ (resp. \(H^\top  \in \mathcal{PF}\))) if $\exists \,\lambda\in\Lambda( H )$ real and positive which is a simple and strictly dominant eigenvalue, and the dominant right (resp. left) eigenvector is positive.
\end{itemize}

\begin{itemize}
    \item The matrix $H \in \mathbb{R}^{n \times n}$ satisfies the Stochastic PF (SPF) property if $H\in\mathcal{PF}$, $H\mathds{1}=\mathds{1}$, and \(\rho(H)=1\).
\end{itemize}

\begin{itemize}
\item The matrix $H \in \mathbb{R}^{n \times n}$ satisfies the right (resp. left) signed PF property (denoted $H \in \mathcal{PF}_S$ (resp. $H^\top \in \mathcal{PF}_S$)), if $\exists \,\lambda\in\Lambda( H )$  a simple real positive and strictly dominant eigenvalue, and the dominant right (resp. left) eigenvector $\mathbf{v}$ is such that $\lvert \mathbf{v}\rvert>0$ \cite{Fontan2022,altafini2014predictable}.
\item A matrix $H \in \mathbb{R}^{n \times n}$ satisfies the signed Stochastic PF (SSPF) property if $H\in\mathcal{PF}_S$, \(\rho(H)=1\), and $ H \mathbf{v} = \mathbf{v}$, where the components of $ \mathbf{v} $ are s.t. $ | v_i |=1 $ $ \forall \, i$. 
An equivalent conditions is that $H\in\mathcal{PF}_S$, \(\rho(H)=1\) and $\Xi H\Xi \mathds{1}=\mathds{1}$, where $ \Xi ={\rm diag}\left(\mathbf{v}\right)$ is also called a gauge transformation.
\end{itemize}

\section{Problem Formulation}
\label{Problemformulation}
Consider a group of \(n\) individuals engaged in a discussion on a specific topic, during which they exchange information and update their opinions according to an interaction graph. Each individual \(i\) starts with an initial opinion, \(x_i(0)\),  drawn independently from a random distribution with mean $ \zeta$ and variance \(\mathrm{Var}[x_i(0)]=\sigma_i^2\) where $0<\sigma_i^2<\infty$. 
Here, \(\zeta\) represents the true underlying value of the topic. A smaller variance \(\sigma_i^2\) indicates that the $i$th individual is an expert with more reliable knowledge, whereas a larger variance reflects less expertise or naivety.
Let \(\mathbf{x}(0)\) denote the vector of $n$ initial opinions. 
The collective wisdom of the group at  $ t=0$ can be quantified by the group mean, i.e., the expected value of the average with respect to the stochasticity in the initial conditions, $ \mathbb{E}[\overline{x}(0)] =\zeta$, where $ \overline{x}(0) = \frac{1}{n} \sum_{i=1}^n x_i(0) $, and the group variance is $ {\rm Var}[\overline{x}(0) ]=\mathds{1}^\top  \Sigma\mathds{1} /{n^2} =  \frac{1}{n^2} \sum_{i=1}^n \sigma_i^2 $ where \(\Sigma={\rm diag}([\sigma_1^2,\,\dots,\,\sigma_n^2])\). 

The discussion process is modeled as a convergent linear dynamical model (more details given below), and we denote $ \mathbf{x}^\ast $ the asymptotic value reached by the opinions. 
For each of the dynamical models we consider below, it is possible to define a notion of social power, representing the influence each agent has on the outcome of the opinion process. Let $ \mathbf{y} = \begin{bmatrix} y_1& \ldots & y_n \end{bmatrix}^\top  $ denote the social power vector, normalized so that $  \mathds{1}^\top  \mathbf{y} =1$. 
In particular, for our models, social powers determine the weights by which the initial opinions get updated through the opinion dynamics processes. More specifically, our opinion dynamics models consist of averaging-like processes. At steady state, their endpoint is $\mathbf{x}^\ast$, from which we can compute the ``sample'' mean $ \overline{x}^\ast =\frac{1}{n}\sum_{i=1}^n x_i^\ast$, which sometimes will coincide with the consensus value achieved by the group, and some other times not (namely, when the model includes stubbornness). In all cases, this can be expressed as a function of $ \mathbf{y}$: $  \overline{x}^\ast = \frac{1}{n}\sum_{i=1}^n x_i^\ast = \mathds{1}^\top  \mathbf{x}^\ast/n =\mathbf{y}^\top \mathbf{x}(0)$. For this $ \overline{x}^\ast  $, we can then compute the group mean (w.r.t. the stochasticity in the initial conditions) $ \mathbb{E}[\overline{x}^\ast]  $ and the group variance $ {\rm Var}[\overline{x}^\ast ]= \sum_{i=1}^n y_i^2 \sigma_i^2 =  \mathbf{y}^\top  \Sigma  \mathbf{y} $. 

To quantify the wisdom of crowd of our dynamical models, we introduce the following definitions. 
\begin{definition}
An opinion dynamics model is said 
\bite
\item \emph{Mean accurate} if $ \mathbb{E}[\overline{x}^\ast] = \zeta$;
\item \emph{Concentrating} if \(\mathrm{Var}[\overline{x}^\ast]<\mathrm{Var}[\overline{x}(0)]\);
\item \emph{Dispersing} if  \(\mathrm{Var}[\overline{x}^\ast]>\mathrm{Var}[\overline{x}(0)]\);
\item \emph{Neutral} if  \(\mathrm{Var}[\overline{x}^\ast]=\mathrm{Var}[\overline{x}(0)]\). 
\eite
It is said to optimize concentration if \(\mathrm{Var}[\overline{x}^\ast]=\min\limits_{\mathbf{y}\in\mathbb{R}^n,\,\mathbf{y}^\top \mathds{1}=1} \mathbf{y}^\top \Sigma\mathbf{y}\).
\end{definition}
In the literature, a mean accurate model which is concentrating is also said to be ``improving the wisdom of crowd'', while a mean accurate model which is dispersing is said to be ``undermining the wisdom of crowd'', see e.g. \cite{tian2023dynamics,tian2023social}. A mean accurate model which is neutral is said to be ``preserving the wisdom of crowd''.

When $ {\bm \sigma} =\{ \sigma_1, \ldots, \sigma_n \}$ is given, we can intend the group variance $ \mathrm{Var}[\overline{x}^\ast] $ as a function of the social power $ \mathbf{y}  $, and denote it $ \mathcal{T}_{{ \bm \sigma}} (\mathbf{y}) = \sum_{i=1}^n y_i^2 \sigma_i^2 $. 
In a similar way, for the initial condition $ \mathbf{x}(0) $ (when no information on the dynamical influence process is available, and the best thing we can do is to put all equal social powers, $ \mathbf{y}_0 = \mathds{1}/n$), we can consider the function $ \mathcal{T}_{{ \bm \sigma}} (\mathds{1}/n) $. 
Passing from $ \overline{x}(0) $ to $\overline{x}^\ast $, concentration occurs if it is $ \mathcal{T}_{{ \bm \sigma}} (\mathbf{y}) < \mathcal{T}_{{ \bm \sigma}} (\mathds{1}/n) $, while dispersion occurs if $ \mathcal{T}_{{ \bm \sigma}} (\mathbf{y}) > \mathcal{T}_{{ \bm \sigma}} (\mathds{1}/n) $. 
Following \cite{tian2023social}, we can capture this condition compactly by introducing the hyperellipsoid $  \mathcal{T}_{{ \bm \sigma}} (\mathbf{y})/ \mathcal{T}_{{ \bm \sigma}} (\mathds{1}/n)= \frac{n^2\sum_{i=1}^n y_i^2 \sigma_i^2}{\sum_{j=1}^n \sigma_j^2}=1$.
The shape of the hyperellipsoid \( {n^2\sum_{i=1}^n  y_i^2 \sigma_i^2 }=\sum_{j=1}^n \sigma_j^2\) depends on the variance terms \(\sigma_i^2\). If all variances are equal, i.e., \(\sigma_i=\sigma_j\) \(\forall\, i,j\), then the equation reduces to \(\sum_{i=1}^n {n} y_i^2=1\) which represents a hypersphere centered at the origin with radius \({1}/{n}\). If the variances \(\sigma_i^2\) vary among individuals, the hyperellipsoid stretches more along the directions where the corresponding \(\sigma_i^2\)  values are smaller.
In other words, the principal axes of the ellipsoid are proportional to the accuracies $ {1}/{\sigma_i^2 }$.
Because of the normalization $ \mathds{1}^\top  \mathbf{y} =1$, when $ \mathbf{y}>0 $ the social power $ \mathbf{y} $ is living on the \(n\)-simplex \(\Delta\). 
The concentration region becomes then the interior of a convex region \(\Gamma_1\) defined as the intersection of \(\Delta\) and the  hyperellipsoid: $ \Gamma_1= \Delta \cap \Phi_1$ where \(\Phi_1=\left\{ {n^2\sum_{i=1}^n  y_i^2 \sigma_i^2}\leq \sum_{j=1}^n \sigma_j^2 \right\}\).

\section{Wisdom of crowd on unsigned networks}
\label{section:wisdomunsignedmodels}

In this section we consider three standard models for opinion dynamics. For all of them we assume that the agents collaborate, which is encoded as an interaction graph $ W\geq 0$.  

\noindent {\bf 1): DeGroot Model.}
The DeGroot model \cite{degroot1974reaching} is the standard model for achieving consensus through averaging. It can be written as
\begin{equation}
    \label{DeGrootModel}
    \mathbf{x}(k+1) =  W\mathbf{x}(k).
\end{equation}
where  \( \mathbf{x}(k) = [x_i(k)]_{n\times 1} \) is the opinion vector at time $k$. As mentioned above,  for the initial conditions we have \(\mathbb{E}[x_i(0)]=\zeta\) and \(\mathrm{Var}[x_i(0)]=\sigma_i^2\). The following assumption is made on the matrix \(W\).
\begin{assumption}
    \label{assumption_DeGroot}
    \(W\) is row-stochastic with \(\rho(W-\mathds{1}\mathbf{z}^\top)<1\) where \(W^\top\mathbf{z}=\mathbf{z}\).
\end{assumption}

\noindent{\bf 2): Friedkin-Johnsen Model.}
In the Friedkin-Johnsen (FJ) model \cite{friedkin1997social}, in addition to averaging, each individual \(i\) has some attachment to its initial opinion during the discussion, represented by a stubbornness parameter \( \theta_i \in (0,1) \). The resulting model is given by:
\begin{equation}
    \label{FJmodel}
     \mathbf{x}(k+1) = (I - \Theta) W\mathbf{x}(k) + \Theta\mathbf{x}(0),
\end{equation}
where \( \Theta = {\rm diag}([\theta_1,\dots,\theta_n]) \) is a diagonal matrix with stubbornness coefficients as diagonal entries. As before, we assume \(\mathbb{E}[x_i(0)]=\zeta\) and \(\mathrm{Var}[x_i(0)]=\sigma_i^2\). In addition, the stubbornness coefficients satisfy the following.
\begin{assumption}
\label{assumption_stubbornness}
    All individuals in the network are partially stubborn, i.e. \(0<\theta_i<1\).
\end{assumption}
\noindent{\bf 3): Concatenated Friedkin-Johnsen Model.}
Assume that the individuals are engaged in a series of discussions \(s=1,2,3,\dots\) on a specific topic. The opinion of each individual is \(x_i(s,k)\) where \(s\) indicates the discussion event and \(k\) the time instant in the $s$-th discussion.  The opinions in the discussion \(s\) are modeled using a FJ model, see \cite{bernardo2021achieving,wang2022consensus}:
\begin{equation}
    \label{concatenatedFJmodel_1}
     \mathbf{x}(s,k+1) = (I - \Theta) W\mathbf{x}(s,k) + \Theta\mathbf{x}(s,0).
\end{equation}
At the start of the next discussion, the individuals start with the final opinions of the previous discussion, hence the name concatenated FJ model:
\begin{equation}
    \label{concatenatedFJmodel_2}
     \mathbf{x}(s+1,0) = \mathbf{x}(s,\infty).
\end{equation}
Each individual starts with an opinion \(x_i(1,0)\) for the first discussion, with mean \(\mathbb{E}[x_i(1,0)]=\zeta\) and variance \(\mathrm{Var}[x_i(1,0)]=\sigma_i^2\).

\subsection{Asymptotic behavior}
All three models have been thoroughly studied in the literature. 
Their asymptotic properties are reviewed in the following lemma.

\begin{lemma} \cite{degroot1974reaching,friedkin1990social,wang2022consensus} The following results hold for the convergence of opinion dynamics models under their respective assumptions:
\benu
\item \textbf{DeGroot Model:} If Assumption~\ref{assumption_DeGroot} is satisfied, then the model~\eqref{DeGrootModel} converges to a consensus value \mbox{\(\mathbf{x}^\ast=\mathbf{z}^\top \mathbf{x}(0)\mathds{1},\)}
and the social power vector is \(\mathbf{y} = \mathbf{z}\), the dominant left eigenvector of \(W\), with \(\mathbf{y}^\top \mathds{1}=1\).
\item \textbf{FJ Model:} If Assumptions~\ref{assumption_DeGroot} and~\ref{assumption_stubbornness} are satisfied, then the FJ model~\eqref{FJmodel} converges to \( \mathbf{x}^\ast=P\mathbf{x}(0),\)
where 
\begin{equation}
    P=(I-(I-\Theta)W)^{-1}\Theta
\end{equation}
with \(P\mathds{1}=\mathds{1}\) and the social power vector is \(\mathbf{y}= \mathds{1}^\top  P/n\).
\item \textbf{Concatenated FJ Model:} If Assumptions~\ref{assumption_DeGroot} and~\ref{assumption_stubbornness} are satisfied, then the concatenated FJ model~\eqref{concatenatedFJmodel_1},~\eqref{concatenatedFJmodel_2} is stable in the $k$ scale and converges to
    \begin{align*}
        \mathbf{x}(s+1,\infty)&=P\mathbf{x}(s,\infty) \label{concatenatedfjmodel}\numberthis
    \end{align*}
Eq.~\eqref{concatenatedfjmodel} gives also the  opinion dynamics model in the discussion scale \(s\), and converges to consensus
    \begin{equation*}
        \label{concatenatedFJmodeleqconverges}
        \mathbf{x}^\ast=P^{\infty}\mathbf{x}(1,0)=\mathbf{p}^\top \mathbf{x}(1,0)\mathds{1},
    \end{equation*}
and the social power at the end of the series of discussions is $ \mathbf{y}=\mathbf{p}$, the dominant left eigenvector of \(P\) with \(\mathbf{y}^\top \mathds{1}=1\) and \(\mathbf{y}=(I-\Theta)^{-1}\Theta \mathbf{z}\), where $ \mathbf{z} $ is the dominant left eigenvector of $ W$. 
\eenu  
\end{lemma}

\subsection{Analysis of the wisdom of crowd}

When $ W\geq 0 $, for the three models discussed above the wisdom of crowd problem is investigated in \cite{tian2023dynamics,tian2023social}. We report the main results in next proposition.
In all 3 cases, $ \mathbf{y}$ denotes the social power of the equilibrium point.

\begin{proposition}
\label{proposition_unsigned}
Consider the following cases
\benu
\item The DeGroot model~\eqref{DeGrootModel} satisfies Assumption~\ref{assumption_DeGroot};
\item The FJ model~\eqref{FJmodel} satisfies Assumptions~\ref{assumption_DeGroot} and~\ref{assumption_stubbornness};
\item The concatenated FJ model~\eqref{concatenatedfjmodel} satisfies Assumptions~\ref{assumption_DeGroot} and~\ref{assumption_stubbornness}.
\eenu
In each of the three cases, the model is mean accurate, i.e., $ \mathbb{E}[\overline{x}^\ast] = \zeta$. Furthermore, if \(\mathbf{y}\) is the social power vector, the variance \( \mathrm{Var}[\overline{x}^\ast]=\mathbf{y}^\top \Sigma\mathbf{y}\)
\begin{itemize}
        \item concentrates around \(\zeta\) if \(\mathbf{y}\in \mathrm{int}(\Gamma_1)\).
        \item is neutral around \(\zeta\) if \(\mathbf{y}\in  \partial \Gamma_1\).
        \item disperses around \(\zeta\) if \(\mathbf{y}\notin \Gamma_1 \).
        \item is optimized  around \(\zeta\) if \(\mathbf{y}=\frac{\Sigma^{-1}\mathds{1}}{\mathds{1}^\top \Sigma^{-1}\mathds{1}}.\)
    \end{itemize}
\end{proposition}

The proof of this proposition is available in \cite{tian2023dynamics,tian2023social} or immediately obtainable from it. For instance, mean accuracy of the FJ model follows from the fact that $ P \mathds{1} = \mathds{1}$ implies that the social powers sum to 1. See also the proof of Theorem~\ref{theorem_ind_unipartite} below.

\section{Wisdom of Crowds in Signed Unipartite Networks}
\label{section:wisdomsignedmodels}

Starting from this section, we deal with the linear opinion dynamics models with antagonistic interactions in the network and find conditions under which the models improve or undermine wisdom. Here, the interaction matrix is denoted \(W_s\) and may have negative terms, i.e., \(W_s\gtreqless0\),  where \(w_{s,ij}>0\) represents friendly interactions and \(w_{s,ij}<0\) represent antagonistic interactions.
In this section, we assume that the right dominant eigenvector of $ W_s $ is positive (the vector \(\mathds{1}\)) in all three models we explore, implying that unanimity of opinions is achieved in spite of the antagonistic interactions. 
The signed version of the three models discussed in the previous section are reviewed next.

\noindent {\bf 1): Signed DeGroot Model.}
We consider the ``repelling'' version of the signed DeGroot model \cite{Fontan2022}, which can be written as
\begin{equation}
    \label{SignedDeGrootModel}
    \mathbf{x}(k+1) =  W_s\mathbf{x}(k)
\end{equation}
where the signed interaction pattern \(W_s\) is  characterized by \(W_s\mathds{1}=\mathds{1}\). 
Unlike the unsigned DeGroot model, in which \(W\) is row-stochastic, the eigenvalues of \(W_s\) in the signed case may exceed \(1\), resulting in an unstable system. 
Convergence requires \(1\in\Lambda(W_s)\) to be a simple and strictly dominant eigenvalue of $ W_s$.
Together $ W_s \mathds{1} = \mathds{1} $ and strict dominance of \(1\in\Lambda(W_s)\) guarantee convergence to a consensus point.
Notice that, unlike in \cite{Fontan2022}, we require that the matrix \(W_s\) satisfies the right PF property but not necessarily the left PF property, implying that the dominant left eigenvector \(\mathbf{z}_s\) may contain negative entries, i.e., we admit negative social powers in the model.
These conditions are summarized in the following assumption on the matrix \(W_s\).
\begin{assumption}
    \label{assumption_SignedDeGroot}
    \(W_s\) satisfies the SPF property.
\end{assumption}

\noindent{\bf 2): Signed Friedkin-Johnsen Model.}
In addition to averaging and stubbornness, the signed FJ (SFJ) model includes antagonism in the network. The resulting model is given in \cite{razaq2025signed}:
\begin{equation}
    \label{SignedFJmodel}
     \mathbf{x}(k+1) = (I - \Theta) W_s\mathbf{x}(k) + \Theta\mathbf{x}(0).
\end{equation}
In this model, Assumption~\ref{assumption_SignedDeGroot} is not sufficient for stability, as shown via counterexamples in \cite{razaq2023propagation}. 
To guarantee stability, we make the following stronger assumption.
\begin{assumption}
    \label{assumption_signedfjmodel}
    \(W_s\) satisfies the SPF property and $ \theta_i$ are s.t. \(0\leq\theta_i<1\), \(\forall i\). Furthermore, \(\rho((I-\Theta)W_s)<1\).
\end{assumption}

\noindent{\bf 3): Concatenated Signed Friedkin-Johnsen Model.}
The concatenated SFJ model investigated in \cite{razaq2025signed} extends the concatenated FJ model to include antagonistic interactions. The opinions in the discussion \(s\) are modeled using a SFJ model:
\begin{equation}
    \label{concatenatedsignedFJmodel_1}
     \mathbf{x}(s,k+1) = (I - \Theta) W_s\mathbf{x}(s,k) + \Theta\mathbf{x}(s,0).
\end{equation}
As in the unsigned case, at the beginning of the next discussion, the individuals start with the final opinions of the previous discussion, see~\eqref{concatenatedFJmodel_2}.
In this model, since \(W_s\) has signed entries, the matrix 
\begin{equation}
P_s = (I-(I-\Theta)W_s)^{-1}\Theta
\label{eq:Ps}
\end{equation}
may have spectral radius greater than \(1\) even if \mbox{\(\rho((I-\Theta)W_s)<1\)}, as shown via counterexamples in \cite{razaq2023propagation}.  We need the following assumption for the stability of concatenated SFJ model, which relies on the matrix \(P_s\).
\begin{assumption}
\label{assumption_signedconcatenatedfjmodel}
    \(W_s\) satisfies the SPF property and $ \theta_i $ are s.t. \(0\leq\theta_i<1\), \(\forall i\). Furthermore, \(\rho((I-\Theta)W_s)<1\) and \(P_s\) satisfies the SPF property.
\end{assumption}

\subsection{Asymptotic behavior}
All three models have been thoroughly studied in the literature, see \cite{Fontan2022,razaq2025signed}. 
Their asymptotic properties are reviewed in the following lemma, which extends the results of \cite{Fontan2022,razaq2025signed} to include $ W_s $ matrices that satisfy the SPF property, i.e., that have the PF property only on the right side and not on the left.

\begin{lemma} \cite{Fontan2022,razaq2025signed}\label{lemma_mergedsigneddynamisc} The following results hold for the convergence of opinion dynamics models with antagonistic interactions under their respective assumptions:
\benu
\item \textbf{Signed DeGroot Model:} If Assumption~\ref{assumption_SignedDeGroot} is satisfied, then the model~\eqref{SignedDeGrootModel} converges to the consensus value \(        \mathbf{x}^\ast=\mathbf{z}_s^\top \mathbf{x}(0)\mathds{1},\)
and the social power vector is \(\mathbf{y} = \mathbf{z}_s\), the dominant left eigenvector of \(W_s\).
\item \textbf{Signed Friedkin-Johnsen Model:} If Assumption~\ref{assumption_signedfjmodel} is satisfied, then the SFJ model~\eqref{SignedFJmodel} converges to \(    \mathbf{x}^\ast=P_s\mathbf{x}(0),\)
where $ P_s $ is given in \eqref{eq:Ps}, \(P_s\mathds{1}=\mathds{1}\), and the social power vector is \(\mathbf{y}=   P_s^\top\mathds{1}/n\).
\item \textbf{Concatenated Signed Friedkin-Johnsen Model:} If Assumption~\ref{assumption_signedconcatenatedfjmodel} is satisfied, then the concatenated SFJ model~\eqref{concatenatedsignedFJmodel_1},~\eqref{concatenatedFJmodel_2} is stable in the $k$ scale and converges to
\begin{align*}
    \mathbf{x}(s+1,\infty)&=P_s\mathbf{x}(s,\infty).
    \numberthis\label{concatenatedsfjmodel}
\end{align*}

Eq.~\eqref{concatenatedsfjmodel} gives also the  opinion dynamics model in the discussion scale \(s\), and converges to consensus \(        \mathbf{x}^\ast=P_s^{\infty}\mathbf{x}(1,0)=\mathbf{p}_s^\top \mathbf{x}(1,0)\mathds{1}.
\)
The social power at the end of the series of discussions is $ \mathbf{y}=\mathbf{p}_s$, the dominant left eigenvector of \(P_s\), expressed as \(\mathbf{p}_s=(I-\Theta)^{-1}\Theta \mathbf{z}_s\), where $ \mathbf{z}_s $ is the dominant left eigenvector of $ W_s$. 
\eenu  
In all three cases, \(\mathbf{y}^\top \mathds{1}=1\), even though for some $i$ it could be $ y_i<0$.
\end{lemma}

\subsection{Analysis of the wisdom of crowd}

In this subsection, conditions for improving or undermining the wisdom of crowd on signed graphs are investigated. %
Since the social power vector \(\mathbf{y}\) may have negative components, we need to redefine the concentration region.
In particular, instead of the simplex $ \Delta$, $ \mathbf{y} $ can now lie in the hyperplane \(\Psi_1=\{\mathbf{y}\in\mathbb{R}^n|\mathds{1}^\top \mathbf{y}=1\}\).
Therefore, instead of $ \Gamma_1 $, we have to consider the convex region $\Gamma_2 = \Psi_1 \cap \Phi_1$. 
The point \(\mathds{1}/n\) is always included in this intersection, \(\mathds{1}/n\in\Gamma_2\). 
Since $ \Delta \subset \Psi_1$, it is  \(\Gamma_1\subseteq\Gamma_2\), i.e., the convex region of social power vectors in which the signed opinion dynamic models improves the wisdom of crowd may be larger than that of the unsigned models. 

Similar to the previous section,  the initial conditions are always such that \(\mathbb{E}[x_i(0)]=\zeta\) and \(\mathrm{Var}[x_i(0)]=\sigma_i^2\). 
 We report the main results in the next Theorem.

\begin{theorem}
\label{theorem_ind_unipartite}
    Consider the following cases
\benu
\item The signed DeGroot model~\eqref{SignedDeGrootModel} satisfies Assumption~\ref{assumption_SignedDeGroot};
\item The SFJ model~\eqref{SignedFJmodel} satisfies Assumption~\ref{assumption_signedfjmodel};
\item The concatenated SFJ model~\eqref{concatenatedsfjmodel} satisfies Assumption~\ref{assumption_signedconcatenatedfjmodel}.
\eenu
In each of the three cases, the model is mean accurate, i.e, $ \mathbb{E}[\overline{x}^\ast] = \zeta$. Furthermore, if \(\mathbf{y}\) is the social power vector, the variance \( \mathrm{Var}[\overline{x}^\ast]=\mathbf{y}^\top \Sigma\mathbf{y}\)
\begin{itemize}
        \item concentrates around \(\zeta\) if \(\mathbf{y}\in\mathrm{\Gamma_2}\).
        \item is neutral around \(\zeta\) if \(\mathbf{y}\in\partial\Gamma_2\).
        \item disperses around \(\zeta\) if \(\mathbf{y}\notin \Gamma_2\). 
        \item is optimized  around \(\zeta\) if \(\mathbf{y}=\frac{\Sigma^{-1}\mathds{1}}{\mathds{1}^\top \Sigma^{-1}\mathds{1}} .\)
    \end{itemize}
\end{theorem}
\proof For the three models, we notice that \(\overline{x}^\ast = \mathbf{y}^\top  \mathbf{x}(0),\) where \(\mathbf{y}^\top \mathds{1} = 1\). The expected value of \(\overline{x}^\ast\) is:  \(\mathbb{E}[\overline{x}^\ast] = \mathbb{E}\left[\mathbf{y}^\top  \mathbf{x}(0)\right]=\mathbf{y}^\top  \mathbb{E}\left[\mathbf{x}(0)\right]=\mathbf{y}^\top \mathds{1}  \zeta=\zeta.\) This means that the three models are mean accurate.

For all three models, the variance of the average of the initial opinions is given by \( \mathrm{Var}[\overline{x}(0)] = \mathds{1}^\top \Sigma\mathds{1}/n^2 =  \sum_{i=1}^n \sigma_i^2/n^2\), while the variance at the end of the discussion(s) is \(\mathrm{Var}[\overline{x}^\ast]=\mathbf{y}^\top  \Sigma \mathbf{y}=\sum_{i=1}^n y_i^2 \sigma_i^2\).
This implies that the variance concentrates around $ \zeta$ if  
\(
\sum_{i=1}^n y_i^2 \sigma_i^2 < \sum_{j=1}^n \sigma_j^2/n^2,
\)
which, combined with \(\mathbf{y} \in \Psi_1\) leads to $ \mathbf{y} \in \mathrm{int}(\Gamma_2)$ where the wisdom of the crowd improves. 
Conversely, the wisdom of the crowd is undermined if \(\mathbf{y} \notin \Gamma_2\), indicating that \(\mathrm{Var}[\overline{x}^\ast]\) disperses around \(\zeta\). 

To prove that the wisdom of the crowd is optimized to \(\frac{1}{\mathds{1}^\top \Sigma^{-1}\mathds{1}}\) when \(\mathbf{y} = \frac{\Sigma^{-1}\mathds{1}}{{\mathds{1}^\top \Sigma^{-1}\mathds{1}}}\), we need to solve the constrained optimization problem:
\[
\min_{\mathbf{y}} \quad \mathbf{y}^\top  \Sigma \mathbf{y} \quad \text{subject to} \quad \mathds{1}^\top  \mathbf{y} = 1.
\]
Define the Lagrangian incorporating the constraint: \(\mathcal{L}(\mathbf{y}, \lambda) = \mathbf{y}^\top  \Sigma \mathbf{y} + \lambda \left( 1 - \mathds{1}^\top  \mathbf{y} \right),\)
where \(\lambda\) is the Lagrange multiplier.
To find the minimum, we take the derivative of \(\mathcal{L}\) with respect to \(\mathbf{y}\) and \(\lambda\), and set them to zero:
\begin{itemize}
    \item Gradient with respect to \(\mathbf{y}\):
    \(
    \frac{\partial \mathcal{L}}{\partial \mathbf{y}} = 2\Sigma \mathbf{y} - \lambda \mathds{1} = 0.
    \)
    Rearranging gives:
     \(
\mathbf{y} = \frac{\lambda}{2} \Sigma^{-1} \mathds{1}.
\)

    \item Gradient with respect to \(\lambda\):
    \(
    \frac{\partial \mathcal{L}}{\partial \lambda} = 1 - \mathds{1}^\top  \mathbf{y} = 0.
    \)
    This enforces the constraint:
    \(
    \mathds{1}^\top  \mathbf{y} = 1.
    \)
\end{itemize}
Substituting \(\mathbf{y}\) into the constraint:
\(
\mathds{1}^\top  \left( \frac{\lambda}{2} \Sigma^{-1} \mathds{1} \right) = 1,
\)
which gives
\(
\lambda = \frac{2}{\mathds{1}^\top  \Sigma^{-1} \mathds{1}}.
\)
Substituting \(\lambda\) back into \(\mathbf{y}\) gives
\[
\mathbf{y} = \frac{\frac{2}{\mathds{1}^\top  \Sigma^{-1} \mathds{1}}}{2} \Sigma^{-1} \mathds{1}= \frac{\Sigma^{-1} \mathds{1}}{\mathds{1}^\top  \Sigma^{-1} \mathds{1}}.
\]
To verify that the variance is optimized,
substitute \(\mathbf{y} = \frac{\Sigma^{-1}\mathds{1}}{{\mathds{1}^\top \Sigma^{-1}\mathds{1}}}\) into \(\mathrm{Var}[\overline{x}^\ast]\). The variance becomes:
\[
\mathrm{Var}[\overline{x}^\ast] = \left( \frac{\Sigma^{-1} \mathds{1}}{\mathds{1}^\top  \Sigma^{-1} \mathds{1}} \right)^\top  \Sigma \left( \frac{\Sigma^{-1} \mathds{1}}{\mathds{1}^\top  \Sigma^{-1} \mathds{1}} \right)= (\mathds{1}^\top  \Sigma^{-1} \mathds{1})^{-1}.
\]
This completes the proof.
\qed

\begin{remark}
The set difference $ \Gamma_2 \setminus \Gamma_1 $ corresponds to some agents having negative social power, case that could happen for all three models under the conditions mentioned in Theorem~\ref{theorem_ind_unipartite}.
\end{remark}

\begin{example}
\label{example1}
    Consider a graph of 3 nodes with interaction matrix 
    \[W_s=\begin{bmatrix}
        0.3 & 0.5 & 0.2\\
        -0.5 & 0.9 &0.6\\
        0.9 &0.4 &-0.3
    \end{bmatrix}\]
    which satisfies the SPF property. The social power vector for the DeGroot model is \(\mathbf{y} =\mathbf{z}_s =[-0.117,\, 0.7766,\, 0.3404]^\top\). This indicates that the group as a whole assigns negative social power to individual \(1\). 
    In this case, if the variances satisfy \(\sigma_1^2 > 5.05\sigma_2^2 + 0.048\sigma_3^2\) (the corresponding concentration region $\Gamma_2$ is shown in Figure~\ref{figure_hyperplanehyperellipsoid} (green + blue areas)), the DeGroot model with \(W_s\) will improve the wisdom of the crowd. For instance, \({  \sigma}_i^2=\{6,\, 1,\, 1\}\) will result in the variance of the average of final opinions \(\mathrm{Var}[\overline{x}^\ast]=0.8\) being smaller than the variance of the average of initial opinions \(\mathrm{Var}[\overline{x}(0)]=0.89\), thus improving the wisdom of the crowd.  \qed
\end{example}
Theorem~\ref{theorem_ind_unipartite}  proves that the wisdom of the crowd can be improved even in the presence of antagonistic interactions in the network. This improvement is particularly evident when there are significant differences in opinion variances among individuals, as shown in Example~\ref{example1}. Figure~\ref{figure_hyperplanehyperellipsoid} shows the convex region \(\Gamma_2\) for \(\sigma_i^2=\{6 ,\,1 ,\,1\}\), which extends significantly beyond the \(3\)-simplex.

\begin{figure}
    \centering    \includegraphics[width=0.8\linewidth]{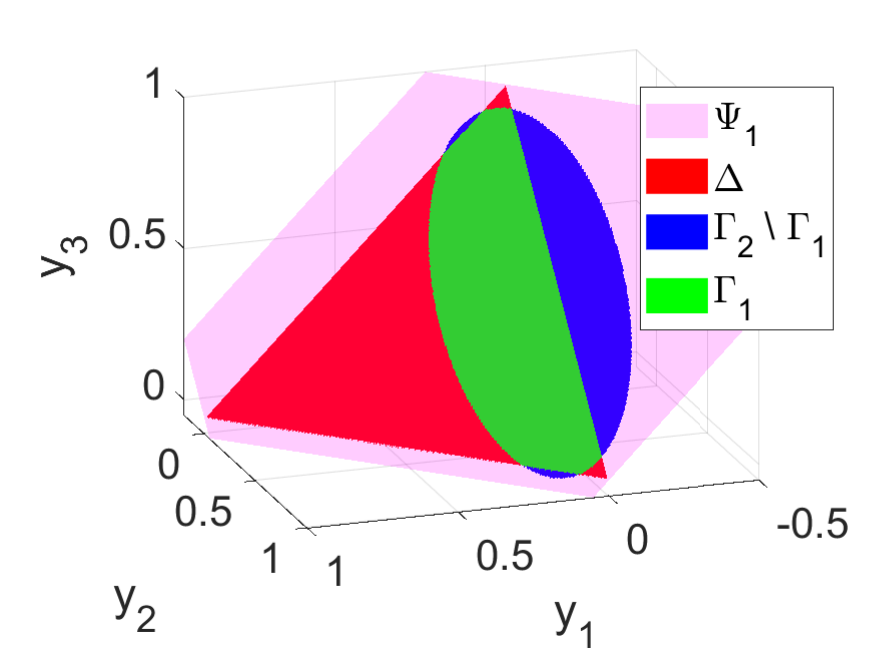}
\caption{Concentration regions $ \Gamma_1 $ and $ \Gamma_2 $ corresponding to an hyperellipsoid determined by the variances $ \sigma_i^2 = \{ 6, \, 1, \, 1\}$. The hyperplane \(\Psi_1\) is shown in pink, the \(3\)-simplex is shown in red, the region $ \Gamma_1 $ is in green while $ \Gamma_2 $ includes both the regions in blue and in green.
}
\label{figure_hyperplanehyperellipsoid}
\end{figure}

A negative social power, i.e., \(y_i<0\) for some $ i$, extends to social powers the idea that agents can be split into two opposite fronts by their antagonistic interaction pattern in $W$. Hence it reflects the group's perception that the individual \(i\)  is exerting a negative influence in the collective judgment, based on the values of the edges of $ W_s$.
When $ W_s $ satisfies the SPF property, the presence of negative entries in the social power vector does not compromise consensus. 

Notice further that \(y_i<0\) is not related to self-weight (and to self-appraisal mechanisms \cite{jia2013dynamics}) in the interaction matrix $ W_s$. 
As seen in Example~\ref{example1}, individual \(1\) has a positive self weight, while individual \(3\) has a negative self weight. However, from the group's perspective, individual \(1\) ends up being assigned a negative social power, whereas individual \(3\) receives a positive social power.

While negative social power reflects the group’s perception of an agent’s influence, negative self-weights in the interaction matrix have a different interpretation. They arise as a structural consequence of enforcing normalization in the presence of strong friendly ties. For instance, in Example~\ref{example1}, node 3 assigns large positive weights to its neighbors (0.9 and 0.4) and a negative weight to itself (\(-0.3\)). This does not imply that the agent ``knows its opinion is wrong,'' but rather that it relies heavily on external opinions and slightly counteracts its previous stance to maintain balance.

Lastly, a negative social power may also result in the consensus value \(\mathbf{x}^\ast\) outside the convex hull of initial opinions, i.e., \(\mathbf{x}^\ast\notin co(\mathbf{x}(0))\) as shown in the examples of \cite{razaq2023propagation,razaq2025signed}. 
Even with this property, the wisdom of the crowd may improve if the conditions of Theorem~\ref{theorem_ind_unipartite} are satisfied.

If also the dominant left eigenvector \(\mathbf{z}_s\) is positive (the matrix is called eventually positive in this case, see \cite{razaq2025signed}), the social powers of all individuals in the signed DeGroot model and in the concatenated SFJ model are positive (for any value of \(\Theta\)). In this case, the regions $\Gamma_1 $ and $ \Gamma_2$ are equal, and Theorem~\ref{theorem_ind_unipartite} is equivalent to  Proposition~\ref{proposition_unsigned}. However, for the SFJ model, the social power may still be negative even when considering an eventually positive matrix \(W_s\), as illustrated in Example~\ref{ESmatrixexample}.
\begin{example}
    \label{ESmatrixexample}
    Consider the eventually positive matrix
    \[W_s=\begin{bmatrix}
        0.94  &  0.76 &  -0.7\\
   -0.06 &   0.61&    0.45\\
    0.25   & 0.9  & -0.15
    \end{bmatrix}\]
    with the stubbornness matrix \(\Theta={\rm diag}(\begin{bmatrix}
        0.1& 0.5& 0.4
    \end{bmatrix})\). The social power vector for the DeGroot model is \(\mathbf{y}=\mathbf{z}_s=\begin{bmatrix}
        0.1220 &   0.6844 &   0.1935    \end{bmatrix}^\top \) and for the SFJ model is \(\mathbf{y}=P^\top \mathds{1}/3=\begin{bmatrix}
            0.1583 &   0.9315 &  -0.0898        \end{bmatrix}^\top, \) while for the concatenated SFJ model, it is \(\mathbf{y}=\mathbf{p}=\begin{bmatrix}
                0.0164  &   0.8276 &  0.1560
            \end{bmatrix}^\top . \)\qed
\end{example}
    In the concatenated SFJ model, we notice that the variance has a transient behavior, where it changes during the initial discussion and settles after several discussions, as shown in Example~\ref{example_fjconcatenatedfj}.
  \begin{example}
  \label{example_fjconcatenatedfj}
    Consider a graph of \(3\) nodes with a SPF interaction  matrix
    \[
        W_s=\begin{bmatrix}
            0.4 & 0.8 & -0.2\\
            0.9 & 0.1 & 0\\
            0.6 & 0.1 & 0.3
        \end{bmatrix}
    \]
If the stubbornness parameters are \(\Theta=\mathrm{diag}(\begin{bmatrix}
        0.8&0.5&0.8
    \end{bmatrix})\), we find that \(\rho((I-\Theta)W_s)<1\). The social power vector for the SFJ model is \(\mathbf{y}=P^\top \mathds{1}/3=\begin{bmatrix}
        0.5057& 0.2321& 0.2622
    \end{bmatrix}^\top ,\) and for the concatenated SFJ model it is \(\mathbf{y}=\mathbf{p}=\begin{bmatrix}
        1.0769& 0.2308&-0.3077
    \end{bmatrix}^\top. \)
    If the initial variances are \(\sigma_i^2=\{1,\, 4,\, 4\}\), the variance of the average of the initial opinions is \(\mathrm{Var}[\overline{{x}}(1,0)]=1\). After the first discussion, the variance decreases to \(\mathrm{Var}[\overline{{x}}(1,\infty)]=0.74\), indicating an improvement in wisdom. However, over the long run, the variance increases to \(\mathrm{Var}[\overline{x}^\ast]=1.75\), showing that while wisdom improves during an intermediate step (in a SFJ model), it is undermined in the long term (in a concatenated SFJ model), see Figure~\ref{figure_Variancepropagationexample2}. The propagation of variance of the average opinion \(\mathrm{Var}[\overline{x}(s,0)]\) is shown in red in Figure~\ref{figure_Variancepropagationexample2}. \qed
\end{example}

A similar transient behavior of the variance can also be observed in the DeGroot and FJ models, as well as within a discussion in the concatenated FJ model when analyzed on the single discussion time scale. In these cases, the variance of the aggregated opinion also has a transient behavior.

\begin{figure}
    \centering    \includegraphics[width=0.6\linewidth]{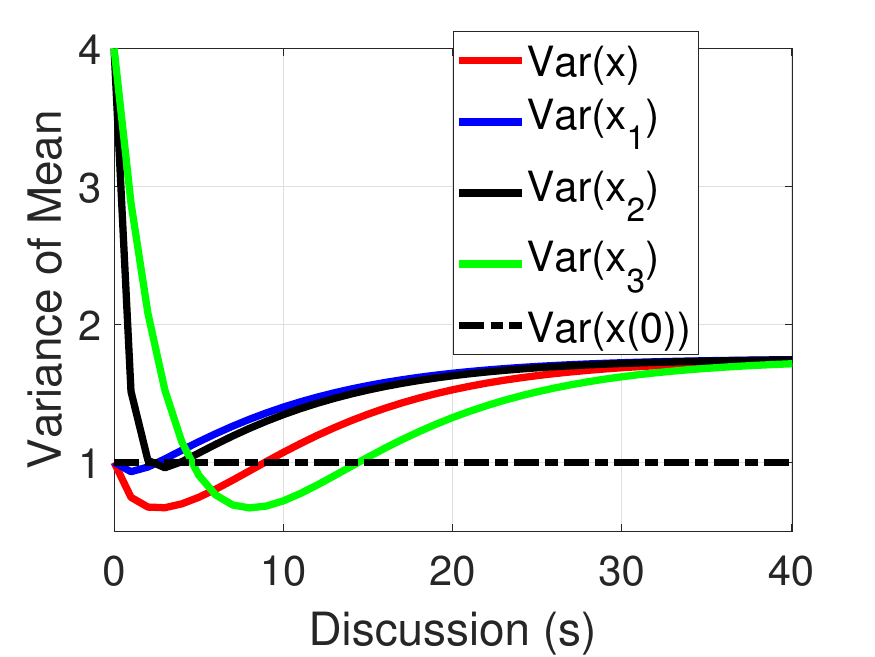}
\caption{Variance propagation for Example~\ref{example_fjconcatenatedfj}}
\label{figure_Variancepropagationexample2}
\end{figure}

\section{Wisdom of Crowd in Signed Bipartite Networks}
\label{section:wisdombipartitemodels}
The signed graphs investigated in the previous section are all characterized by a positive dominant right eigenvector.  
This limits the amount of antagonism that can be encoded in the model and renders the wisdom of crowd behavior similar to the unsigned graph case. 
To go beyond these situations and to treat graphs with an arbitrary amount of antagonism, consider now a signed digraph with a ``right-hand side bipartite structure'', i.e., such that the interaction matrix, call it $ W_b $, has a dominant right eigenvector $ \mathbf{v}$ of mixed signature, that is, $W_b$ satisfies the SSPF property.
Then, we know from the definition of SSPF that there exists a gauge transformation matrix $\Xi={\rm diag} \left(\mathbf{v}\right)$ such that  $W_s=\Xi W_b\Xi$ satisfies the SPF property.
Consider the associated models
\benu
\item Signed DeGroot model:
\begin{equation}
\mathbf{x}\left(k+1\right)=W_b\mathbf{x}\left(k\right)  \label{bipartiteDeGroot} 
\end{equation}
\item SFJ: 
\begin{equation}
\mathbf{x}\left(k+1\right)=\left(I-\Theta\right)W_b\mathbf{x}(k)+\Theta \mathbf{x}(0)  \label{bipartiteSFJ}
\end{equation}
\item Concatenated SFJ:
\begin{equation}
\mathbf{x}(s,k+1)=\left(I-\Theta\right)W_b\mathbf{x}(s,k)+\Theta \mathbf{x}(s,0)\label{bipartiteConcatSFJ1}
\end{equation}
with Eq.~\eqref{concatenatedFJmodel_2} for the concatenation of discussions.
\eenu
It follows by construction that these systems are gauge transformations of the models \eqref{SignedDeGrootModel}, \eqref{SignedFJmodel}, and \eqref{concatenatedsignedFJmodel_1} we considered in the previous section, obtained by simply applying a change of basis $\mathbf{r}(t)=\Xi\, \mathbf{x}(t)$. Consider the following assumptions for the bipartite structure of the network. 

\begin{assumption}
\label{assumption_signedDeGrootbipartite}
    \(W_b\) satisfies the SSPF property.
\end{assumption}
 \begin{assumption}
\label{assumption_signedfjmodelbipartite}
    \(W_b\) satisfies the SSPF property and \(\theta_i\) s.t. \(0\leq\theta_i<1\), \(\forall i\). Furthermore, \(\rho((I-\Theta)W_b)<1\).
\end{assumption}

    \begin{assumption}
\label{assumption_signedconcatenatedfjmodelbipartite}
    \(W_b\) satisfies the SSPF property and \(\theta_i\) s.t. \(0\leq\theta_i<1\), \(\forall i\). Furthermore, \(\rho((I-\Theta)W_b)<1\) and 
    \begin{equation}
    \label{eq:Pb}
        P_b=(I-(I-\Theta)W_b)^{-1}\Theta     
    \end{equation}
    satisfies the SSPF property.
\end{assumption}

\subsection{Asymptotic Behavior}
The peculiarity of the signed bipartite case is that at the equilibrium $ \mathbf{x}^\ast$ both $ \mathbf{v} $ and $ \mathbf{y}$, the right and left dominant eigenvectors of the interaction matrix of relevance,  typically have components of mixed signs (and not necessarily the same). 
For consensus-seeking models like the signed DeGroot and concatenated SFJ models, this implies that bipartite consensus is achieved, as stated in the next lemma (given without proof, as it can be easily derived from Lemma~\ref{lemma_mergedsigneddynamisc} through a gauge transformation).

\begin{lemma} \cite{Fontan2022,razaq2025signed} 
\label{lemma:SPF-asympt} The following results hold for the convergence of bipartite opinion dynamics models under their respective assumptions:
\benu
    \item \textbf{Signed DeGroot Model:} If Assumption~\ref{assumption_signedDeGrootbipartite} is satisfied, then the model~\eqref{bipartiteDeGroot} converges to the consensus value
    \(
       \mathbf{x}^\ast=\mathbf{z}_b^\top \mathbf{x}(0)\mathbf{v},
    \)
where $ \mathbf{v} $ is s.t. $ | v_i| =1$ $ \forall \, i$, and the social power vector is \(\mathbf{y}=\mathbf{z}_b\), the dominant left eigenvector of \(W_b\) with \(\mathbf{y}^\top \mathbf{v}=1\).
    \item \textbf{SFJ Model:} If Assumption~\ref{assumption_signedfjmodelbipartite} is satisfied, then the SFJ model~\eqref{bipartiteSFJ} converges to \(    \mathbf{x}^\ast=P_b\mathbf{x}(0),\)
 where \(P_b\) is defined in Eq.~\eqref{eq:Pb} with \(P_b\mathbf{v}=\mathbf{v}\) and $ \mathbf{v} $ is s.t. $ | v_i| =1$ $ \forall \, i$, and the social power vector is \(\mathbf{y}=P_b^\top \mathds{1}/n\) with \(\mathbf{y}^\top \mathbf{v}=\mathds{1}^\top \mathbf{v}/n\).
    \item \textbf{Concatenated SFJ Model:} If Assumption~\ref{assumption_signedconcatenatedfjmodelbipartite} is satisfied, then the concatenated SFJ model~\eqref{bipartiteConcatSFJ1},~\eqref{concatenatedFJmodel_2} is stable in the $k$ scale and converges to
    \begin{equation}
        \mathbf{x}(s+1,\infty)=P_b\mathbf{x}(s,\infty)\label{bipartiteConcatSFJ}.
    \end{equation}
Eq.~\eqref{bipartiteConcatSFJ} gives also the  opinion dynamics model in the discussion scale \(s\), and converges to consensus
    \(        \mathbf{x}^\ast=P_b^\infty\mathbf{x}(1,0)=\mathbf{p}_b^\top \mathbf{x}(1,0)\mathbf{v}.\)
    The social power vector at the end of the series of discussions is  \(\mathbf{y}=\mathbf{p}_b\), the dominant left eigenvector of \(P_b\) which can be expressed also as \(\mathbf{y}=(I-\Theta)^{-1}\Theta\mathbf{z}_b\) where \(\mathbf{z}_b\) is the left eigenvector of \(W_b\), with \(\mathbf{y}^\top \mathbf{v}=1\), (and $ P_b \mathbf{v} = \mathbf{v}$, of components $ |v_i |=1$ $ \forall \, i$).
\eenu 
\end{lemma}

\subsection{Analysis of the wisdom of crowd}
As in previous sections, we assume that the individuals' initial opinion \(x_i(0)\) have mean \(\mathbb{E}[x_i(0)]=\zeta\) and variance \(\mathrm{Var}[x_i(0)]=\sigma_i^2\). Here, in the bipartite case, the average at the equilibrium point can be computed in different ways, either treating all agents equally $ \overline{x}^\ast =\mathds{1}^\top \mathbf{x}^\ast/n$ (we call it population-based average), or respecting the group bipartition $ \overline{x}^\ast =\mathbf{v}^\top \mathbf{x}^\ast/n$ (below denoted bipartition-based average).
Both cases are discussed in detail below.

\subsubsection{Population-based Average}
We aim to compute $ \mathbb{E}[\overline{x}^\ast]$  and $ \mathrm{Var}[\overline{x}^\ast]$, when $ \overline{x}^\ast =\mathds{1}^\top \mathbf{x}^\ast/n$.
The following Theorem provides sufficient conditions on wisdom for the DeGroot and concatenated SFJ model.
\begin{theorem}
\label{theorem_ind_bipartite_degroot}
   Consider the following cases
\benu
\item The signed DeGroot model~\eqref{bipartiteDeGroot} satisfies Assumption~\ref{assumption_signedDeGrootbipartite};
\item The concatenated SFJ model~\eqref{bipartiteConcatSFJ} satisfies Assumption~\ref{assumption_signedconcatenatedfjmodelbipartite}.
\eenu
In each of the two cases, the model is not mean accurate, unless \(\mathbf{y}^\top \mathds{1}=n(\mathds{1}^\top \mathbf{v})^{-1}\), where \(\mathbf{y}\) is the social power vector: the population-based average is $ \mathbb{E}[\overline{x}^\ast] =\mathds{1}^\top \mathbf{v}\mathbf{y}^\top \mathds{1} \zeta/n$ with \(\mathbf{y}^\top \mathds{1} \ne 1\) and \( |\mathds{1}^\top \mathbf{v}/n|<1\). 
Furthermore, the variance \( \mathrm{Var}[\overline{x}^\ast]=(\mathds{1}^\top \mathbf{v}/n)^{2}\mathbf{y}^\top \Sigma\mathbf{y}\)
\begin{itemize}
        \item concentrates around \(\mathds{1}^\top \mathbf{v}\mathbf{y}^\top \mathds{1} \zeta/n\) if \(\mathbf{y}\in \mathrm{int} (\Gamma_3)\)
        \item is neutral around \(\mathds{1}^\top \mathbf{v}\mathbf{y}^\top \mathds{1} \zeta/n\) if \(\mathbf{y}\in \partial\Gamma_3\)
        \item disperses around \(\mathds{1}^\top \mathbf{v}\mathbf{y}^\top \mathds{1} \zeta/n\) if \(\mathbf{y}\notin \Gamma_3\)
        
        \item is optimized  around \(\mathds{1}^\top \mathbf{v}\mathbf{y}^\top \mathds{1} \zeta/n\) if \(\mathbf{y}=\frac{ \Sigma^{-1}\mathbf{v}}{\mathbf{v}^\top \Sigma^{-1}\mathbf{v}}\) 
    \end{itemize}
    where $ \Gamma_3 = \Psi_2 \cap \Phi_2$, with $\Psi_2=\{\mathbf{y}\in\mathbb{R}^n|\mathbf{y}^\top \mathbf{v}=1\}$ and  $ \Phi_2 =\left\{\sum_{i=1}^n y_i^2\sigma_i^2 \leq  (\mathds{1}^\top \mathbf{v})^{-2} \sum_{i=1}^n\sigma_i^2\right\}$.
\end{theorem}
\proof To show failure in mean accuracy, consider for instance the signed DeGroot case.  The right eigenvector $ \mathbf{v} $ associated to the simple and strictly dominant eigenvalue \(1\in\Lambda(W_b)\) has components of mixed sign. In this case, \(W_b^\infty\) exists and we have \(W_b^\infty=\mathbf{v}\mathbf{y}^\top \), where the left eigenvector $ \mathbf{y}$ (denoted $ \mathbf{z}_b $ in the statement of Lemma~\ref{lemma:SPF-asympt}) has also components of mixed signs (and not necessarily the same sign pattern as $ \mathbf{v}$).  
Since, \(W_b^\infty\mathbf{v}=\mathbf{v}\) we have \(\mathbf{v}\mathbf{y}^\top \mathbf{v}=\mathbf{v}\) which results in \(\mathbf{y}^\top \mathbf{v}=1\) or \(\sum_{i=1}^n y_iv_i=1\).
If \(\mathbf{v}\) is not strictly positive, then this implies $ \mathds{1}^\top  \mathbf{y} \neq 1 $. 
As \(\mathbf{x}^\ast=W_b^\infty \mathbf{x}(0)\), it is \(\mathbf{x}^\ast=\mathbf{v}\mathbf{y}^\top  \mathbf{x}(0)\) and \(\overline{x}^\ast=\mathds{1}^\top \mathbf{v}\mathbf{y}^\top \mathbf{x}(0)/n\) with \(|\mathds{1}^\top \mathbf{v}|<n\).
Finally, \(\mathbb{E}[\mathbf{x}(0)]=\zeta \mathds{1}\) implies \(\mathbb{E}[\overline{x}^\ast]=\mathds{1}^\top \mathbf{v}\mathbf{y}^\top \mathds{1}\zeta/n\ne\zeta\), i.e., the mean is not accurate, if \(n (\mathds{1}^\top \mathbf{v})^{-1}\ne\mathbf{y}^\top \mathds{1}\).
The argument is similar for the concatenated SFJ model. 

Concerning the variance, note that the social power vector $\mathbf{y}$ lies in the hyperplane \(\Psi_2\). 
At equilibrium, the population-based average becomes
\(
\overline{x}^\ast = \frac{\mathds{1}^\top \mathbf{v}}{n}\,\mathbf{y}^\top \mathbf{x}(0),
\)
so its variance is
\[
\mathrm{Var}[\overline{x}^\ast] = \left(\frac{\mathds{1}^\top \mathbf{v}}{n}\right)^2 \sum_{i=1}^n y_i^2 \sigma_i^2.
\]
The wisdom will improve if
\[
\left(\frac{\mathds{1}^\top \mathbf{v}}{n}\right)^2 \sum_{i=1}^n y_i^2 \sigma_i^2 < \frac{1}{n^2}\sum_{i=1}^n \sigma_i^2,
\]
which is equivalent to
\(
\sum_{i=1}^n y_i^2\sigma_i^2 <  (\mathds{1}^\top \mathbf{v})^{-2} \sum_{i=1}^n\sigma_i^2,\)
that is the interior of hyperellipsoid \(\Phi_2\).
Then, the concentration region is the interior of the convex set
\(
\Gamma_3 = \Psi_2 \cap \Phi_2.
\)
Thus, $\mathbf{y} \in \mathrm{int}(\Gamma_3)$ implies variance concentration, while $\mathbf{y} \notin \Gamma_3$ implies dispersion.

Finally, the optimization problem $\min\limits_{\mathbf{y}^\top \mathbf{v}=1}\mathbf{y}^\top \Sigma \mathbf{y}$ is a strictly convex quadratic program with linear constraint; the unique minimizer is obtained by the Lagrangian
\(
\mathcal{L}(\mathbf{y},\lambda)=\mathbf{y}^\top\Sigma \mathbf{y}+\lambda(1-\mathbf{y}^\top \mathbf{v})
\),
yielding $2\Sigma \mathbf{y}-\lambda \mathbf{v}=0$. Hence the optimal $ \mathbf{y}$ is 
\(
\mathbf{y}^\ast =\frac{\lambda}{2} \Sigma^{-1}\mathbf{v}
\).
Imposing $\mathbf{y}^{\ast\top}\mathbf{v}=1$ gives $\lambda=2/( \mathbf{v}^\top \Sigma^{-1}\mathbf{v})$ and therefore
\[
\mathbf{y}^\ast \;=\; \frac{1}{\mathbf{v}^\top \Sigma^{-1}\mathbf{v}}\,\Sigma^{-1}\mathbf{v},
\qquad
\mathbf{y}^{\ast\top}\Sigma \mathbf{y}^\ast \;=\; \frac{1}{\mathbf{v}^\top \Sigma^{-1}\mathbf{v}}.
\]
Lastly, multiplying by the prefactor $(\mathds{1}^\top \mathbf{v}/n)^2$ gives the optimal population-mean variance for the two models. \qed

\begin{remark}
\label{remark_bipartition}
Notice that as the two sides of the opinion bipartition approach each other in size (i.e., the number of individuals having opposite opinions becomes more similar), the quantity \(\mathds{1}^\top \mathbf{v}/n\) decreases. This scaling factor appears in the expression for the variance at equilibrium:
\(
\mathrm{Var}[\overline{x}^\ast] = \left({\mathds{1}^\top \mathbf{v}}/{n}\right)^2 \mathbf{y}^\top \Sigma \mathbf{y}.
\)
Although \(\mathbf{y}^\top \Sigma \mathbf{y}\) remains unchanged under sign flips in \(\mathbf{y}\), the term \((\mathds{1}^\top \mathbf{v}/n)^2\) becomes smaller as the bipartition becomes more symmetric. Consequently, the size of the hyperellipsoid \(\Phi_2\) increases and the convex region \(\Gamma_3\) expands, leading to a decrease in the final variance, i.e, to a greater concentration effect. Moreover, when the sizes become equal i.e., \(\mathds{1}^\top \mathbf{v}=0\), both mean and variance vanish: \( \mathbb{E}[\overline{x}^\ast] =0\) and \(\mathrm{Var}[\overline{x}^\ast]=0\). 
\end{remark}

The following Theorem provides sufficient conditions on wisdom for the SFJ model.

\begin{theorem}
\label{theorem_ind_bipartite_SFJ}
If Assumption~\ref{assumption_signedfjmodelbipartite} is satisfied, then the SFJ model~\eqref{bipartiteSFJ} is not mean accurate: if \(\mathbf{y}\) is the social power vector, the population-based average is $ \mathbb{E}[\overline{x}^\ast] =\mathbf{y}^\top \mathds{1} \zeta$, with \(\mathbf{y}^\top \mathds{1} \ne 1\).
Furthermore, the variance \( \mathrm{Var}[\overline{x}^\ast]=\mathbf{y}^\top \Sigma\mathbf{y}\)
\begin{itemize}
        \item concentrates around \(\mathbf{y}^\top \mathds{1} \zeta\) if \(\mathbf{y}\in \mathrm{int}(\Gamma_4)\);
        \item is neutral around \(\mathbf{y}^\top \mathds{1} \zeta\) if \(\mathbf{y}\in \partial\Gamma_4\);
        \item disperses around \(\mathbf{y}^\top \mathds{1} \zeta\) if \(\mathbf{y}\notin \Gamma_4\);
        \item is optimized  around \(\mathbf{y}^\top \mathds{1} \zeta\) if \(\mathbf{y}=\frac{\mathds{1}^\top \mathbf{v}}{n}\frac{\Sigma^{-1}\mathbf{v}}{\mathbf{v}^\top \Sigma^{-1}\mathbf{v}} \), provided that \(\mathds{1}^\top \mathbf{v}\ne 0\),
    \end{itemize}
    where $ \Gamma_4 = \Psi_3 \cap \Phi_1$, with \(\Psi_3=\{\mathbf{y}\in\mathbb{R}^n|\mathbf{y}^\top \mathbf{v}=\mathds{1}^\top \mathbf{v}/n)\}\).
\end{theorem}
The proof follows a similar approach to that of Theorem~\ref{theorem_ind_unipartite} and~\ref{theorem_ind_bipartite_degroot}.
For instance, to prove that the bipartite SFJ model is not mean accurate, we notice  that using a gauge transformation, \[\mathbf{y}^\top \mathds{1}  = \frac{\mathds{1}^\top P_b\mathds{1}}{n}= \frac{\mathds{1}^\top \Xi P_s \Xi\mathds{1}}{n}\ne 1.\] This inequality always holds unless \(\Xi=I\), \(\frac{\mathds{1}^\top  P_s \mathds{1}}{n}= 1\) which is the case for the unipartite SFJ model treated in Theorem~\ref{theorem_ind_unipartite}. The optimal social power vector is scaled by the factor \(\mathds{1}^\top \mathbf{v}/n\) so that it satisfies the affine constraint \(\mathbf{y}^\top\mathbf{v}=\mathds{1}^\top \mathbf{v}/n\), thereby ensuring that \(\mathbf{y}\)  lies in the hyperplane \(\Psi_3\).

\begin{example}
  \label{example_signedbipartitedegroot}
    Consider a graph of \(3\) nodes with the SSPF interaction matrix 
    \[W_b=\begin{bmatrix}
        0.3&-0.6&-0.1\\-0.3&0.8&-0.1\\-0.2&0.9&-0.1
    \end{bmatrix}.\]
    For this \(W_b\), the social power vector for the signed DeGroot model is \(\mathbf y^\top =\mathbf{z}_b^\top =[0.3039  ,\,  -0.7353   ,\, 0.0392]\) with the dominant right eigenvector \(\mathbf{v}^\top =[1 ,\,-1,\, -1]\). If the initial opinions of the individuals \(\mathbf{x}(0)\) are selected from a random distribution with mean \(\zeta=5\) and variances \(\sigma_i^2=\{4,\,1,\,8\}\), the variance of the average opinion at the start of the discussion is \(\mathrm{Var}[\overline{x}(0)]=1.44\). 
    
Through the DeGroot dynamics, the population-based average of the individuals concentrates around the false truth \(\mathbb{E}[\overline{x}^\ast]=\mathds{1}^\top \mathbf{v}\mathbf{y}^\top \mathds{1}\zeta/n=0.6536\), with variance \(\mathrm{Var}[\overline{x}^\ast]=0.1025\). If instead we consider the unipartite network, i.e., \(W_s=\Xi W_b\Xi\) with \(\Xi={\rm diag}(\mathbf{v})\), then the opinions concentrate around the true value \(\mathbb{E}[\overline{x}^\ast]=5\) with the variance \(\mathrm{Var}[\overline{x}^\ast]=0.9224\). 

    Now, consider the SFJ case. For stubbornness values \(\Theta={\rm diag}([0.2,\,0.4,\,0.6])\), the social power vector for the SFJ model is \(\mathbf{y}^\top =\mathds{1}^\top  P_b/3=[0.0550 ,\,   0.2286 ,\,   0.1598]\).
    Following the SFJ model, the population average of the individuals concentrate around the false truth \(\mathbb{E}[\overline{x}^\ast]=\mathbf{y}^\top \mathds{1}\zeta=2.2170\), with variance \(\mathrm{Var}[\overline{x}^\ast]=0.2686\). If instead we consider the unipartite SFJ model, i.e., \(W_s=\Xi W_b\Xi\) with same stubbornness, then the opinions concentrate around the true value \(\mathbb{E}[\overline{x}^\ast]=5\) with the variance \(\mathrm{Var}[\overline{x}^\ast]=0.7901\).

    Lastly, considering the concatenated SFJ case with the same stubbornness values, the social power vector for the concatenated SFJ case is \(\mathbf{y} =\mathbf{p} =[0.1498  ,\, -0.9662  ,\,  0.1159]^\top\).
    For the concatenated SFJ model, the population average of the individuals concentrate around the false truth \(\mathbb{E}[\overline{x}^\ast]=\mathds{1}^\top \mathbf{v}\mathbf{y}^\top \mathds{1}\zeta/n=1.1675\), with variance \(\mathrm{Var}[\overline{x}^\ast]=0.1256\). If instead we consider the unipartite concatenated SFJ model, i.e., \(W_s=\Xi W_b\Xi\) with the same stubbornness, then the opinions concentrate around the true value \(\mathbb{E}[\overline{x}^\ast]=5\) with the variance \(\mathrm{Var}[\overline{x}^\ast]=1.1308\). The variance propagation through the discussions for the bipartite case is shown in Figure~\ref{fig:variancepropagation_bipartite}.
    \qed
\end{example}
\begin{figure}
    \centering
    \includegraphics[width=0.6\linewidth]{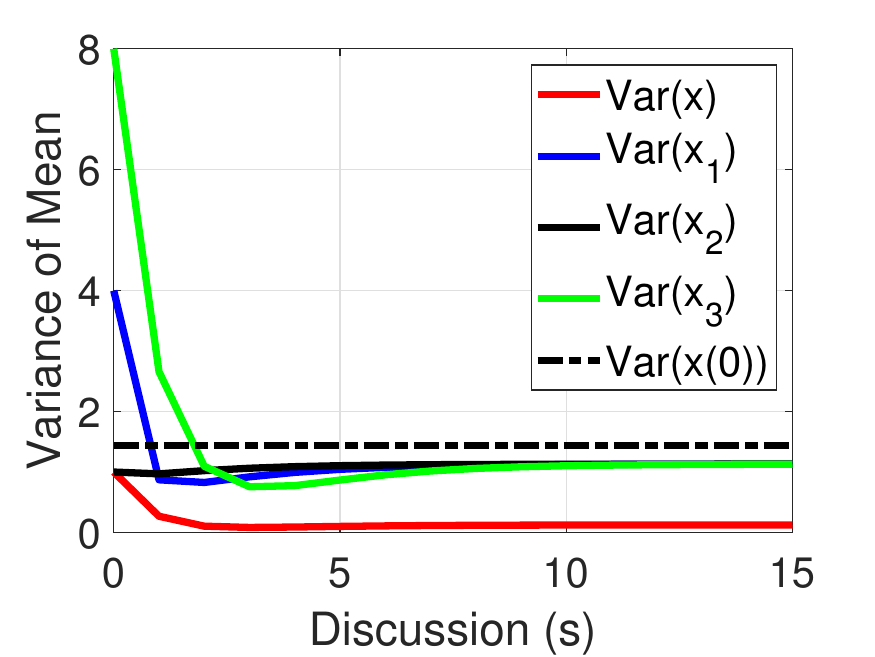}
    \caption{Variance propagation (population-based average case) for Example~\ref{example_signedbipartitedegroot}}
    \label{fig:variancepropagation_bipartite}
\end{figure}

In the bipartite concatenated SFJ case, both the mean and variance change throughout the discussions. In contrast, in the unipartite case, the mean consistently remains equal to the true value, while the variance fluctuates during the discussions. Additionally, in the unipartite case, the collective variance at the end of the discussions matches the individual variances as seen in Figure~\ref{figure_Variancepropagationexample2}. 
On the contrary, in the bipartite case, the collective variance differs from the individual variances, converging to a smaller value  as seen in Figure~\ref{fig:variancepropagation_bipartite}.

\subsubsection{Bipartition-Based Average}

In a signed $ \mathbf{v}$, the sign pattern of the  component \( v_i = \pm 1 \) determine the bipartition of the individuals into two distinct groups (visible in the bipartite consensus in the DeGroot and concatenated FJ models). 
When computing a group average, rather than simply averaging over all individuals, we can define the mean of individuals' opinions \(\overline{x}^\ast\) based on this bipartition as
\(
\overline{x}^\ast = \mathbf{v}^\top \mathbf{x}^\ast/n
\).

In this bipartition-based average case, the definition of the social power vector changes for the SFJ model. Specifically,  the social power vector becomes \(\mathbf{y} = {P_b^\top \mathbf{v}}/{n}.\)
The change reflects the fact that the averaging operator now weights opinions according to group affiliation rather than uniformly across all individuals.
Using the bipartition-based mean \(
\overline{x}^\ast = \mathbf{v}^\top \mathbf{x}^\ast/n\), we get the following Theorem.
\begin{theorem}
\label{theorem_ind_bipartite_group}
   Consider the following cases
\benu
\item The signed DeGroot model~\eqref{bipartiteDeGroot} satisfies Assumption~\ref{assumption_signedDeGrootbipartite};
\item The SFJ model~\eqref{bipartiteSFJ} satisfies Assumption~\ref{assumption_signedfjmodelbipartite}.
\item The concatenated SFJ model~\eqref{bipartiteConcatSFJ} satisfies Assumption~\ref{assumption_signedconcatenatedfjmodelbipartite}.
\eenu
In each of the three cases, the model is not mean accurate: if \(\mathbf{y}\) is the social power vector, the bipartition-based mean is $ \mathbb{E}[\overline{x}^\ast] =\mathbf{y}^\top \mathds{1} \zeta\ne \zeta$ with  \( \mathbf{y}^\top \mathbf{v}=1\) and \(\mathbf{v}\ne\mathds{1}\). 
Furthermore, the variance \( \mathrm{Var}[\overline{x}^\ast]=\mathbf{y}^\top \Sigma\mathbf{y}\)
\begin{itemize}
        \item concentrates around \(\mathbf{y}^\top \mathds{1} \zeta\) if \(\mathbf{y}\in \mathrm{int}(\Gamma_5)\)
            \item is neutral around \(\mathbf{y}^\top \mathds{1} \zeta\) if \(\mathbf{y}\in \partial\Gamma_5\)
        \item disperses around \(\mathbf{y}^\top \mathds{1} \zeta\) if \(\mathbf{y}\notin \Gamma_5\)
        \item is optimized  around \(\mathbf{y}^\top \mathds{1} \zeta\) if \(\mathbf{y}=\frac{\Sigma^{-1}\mathbf{v}}{\mathbf{v}^\top \Sigma^{-1}\mathbf{v}} \) 
    \end{itemize}
    where $ \Gamma_5 = \Psi_2 \cap \Phi_1$.
\end{theorem}
\begin{proof}
For the signed DeGroot and concatenated SFJ models, the equilibrium satisfies
\(
\mathbf{x}^\ast = \mathbf{v}\,\mathbf{y}^\top \mathbf{x}(0),
\)
with  $\mathbf{y}^\top \mathbf{v} = 1$. Thus,
\[
\overline{x}^\ast = \frac{\mathbf{v}^\top \mathbf{x}^\ast}{n} = \frac{\mathbf{v}^\top \mathbf{v}}{n}\,\mathbf{y}^\top \mathbf{x}(0) = \mathbf{y}^\top \mathbf{x}(0),
\]
since $\mathbf{v}^\top \mathbf{v} = n$. For the SFJ model, we have \(\mathbf{x}^\ast=P_b\mathbf{x}(0)\) which leads to $\overline{x}^\ast = \mathbf{v}^\top P_b \mathbf{x}(0)/n$. This also implies \(\overline{x}^\ast =  \mathbf{y}^\top \mathbf{x}(0)\) with \(\mathbf{y}=P_b^\top\mathbf{v}/n\). Taking expectations with $\mathbb{E}[\mathbf{x}(0)] = \zeta \mathds{1}$ gives
\(
\mathbb{E}[\overline{x}^\ast] = \mathbf{y}^\top \mathds{1}\,\zeta,
\)
which equals $\zeta$ only if $\mathbf{y}^\top \mathds{1} = 1$. In the bipartite case, this condition does not  hold because $\mathbf{y}$ is constrained by $\mathbf{y}^\top \mathbf{v} = 1$, so the model is not mean accurate.

For the variance, we have
\(
\mathrm{Var}[\overline{x}^\ast] =\sum_{i=1}^n y_i^2 \sigma_i^2.
\)
So, the condition for concentration is
\(
\sum_{i=1}^n y_i^2 \sigma_i^2 < \frac{1}{n^2}\sum_{i=1}^n  \sigma_i^2,
\)
under the constraint $\mathbf{y}^\top \mathbf{v} = 1$. This defines the interior of the convex region $\Gamma_5$.
Finally, computing the \(\mathbf{y}^\ast\) follows the same approach of the proof of Theorem~\ref{theorem_ind_unipartite}.\qed
\end{proof}

\begin{example}[Example~\ref{example_signedbipartitedegroot} continued]
    \label{example_bipartite_two}
If we consider \(\overline{x}^\ast=\mathbf{v}^\top \mathbf{x}^\ast/n\) with the dominant right eigenvector \(\mathbf{v}^\top =[1 ,\,-1,\, -1]\), then  the social power vector for the DeGroot model is \(\mathbf y=\mathbf{z}_b =[0.3039  ,\,  -0.7353   ,\, 0.0392]^\top \). Also, in the bipartition-based mean case, the opinions reach a false truth \(\mathbb{E}[\overline{x}^\ast]=1.961\) with variance \(\mathrm{Var}[\overline{x}^\ast]=0.9224\). 
Notice that in this case the variance is equal to the one obtained in the unipartite case when \(W_s=\Xi W_b\Xi\) whose social power is \(\mathbf{y}=\mathbf{z}_s =[0.3039  ,\,  0.7353   ,\, -0.0392]^\top\). The latter however is mean accurate.   \qed
\end{example}

\section{Wisdom of Crowd for Dependent Opinions}
\label{Section:dependent}

So far we have only considered the case of independent initial opinions.
In this section, we extend the results obtained in previous sections to the more general situation in which the initial opinions of the individuals are interdependent. Specifically, the initial opinion \(x_i(0)\) of each individual \(i\) is drawn from a distribution with mean \(\mathbb{E}[x_i(0)]=\zeta\) and variance \(\mathrm{Var}[x_i(0)]=\sigma_i^2\), while also incorporating covariances such that \(\mathrm{Cov}[x_i(0)x_j(0)]=\rho_{ij}\sigma_{i}\sigma_{j}\)  where \(-1<\rho_{ij}<1\) is the correlation coefficient. 
Denoting $ \Sigma $ the covariance matrix, of entries \(\Sigma_{ii}=\sigma_i^2\) and \(\Sigma_{ji}=\Sigma_{ij}=\rho_{ij}\sigma_{i}\sigma_{j}\), it is well known that \(\Sigma\) must be positive semi-definite (p.s.d), i.e., \(\Sigma\succeq 0\).

Observe that when \(\Sigma\succeq 0\), the variance of the average of the initial opinion becomes 
\[\mathrm{Var}[\overline{x}(0)]=\mathds{1}^\top \Sigma\mathds{1}/n^2=\frac{1}{n^2}\left(\sum\limits_{i=1}^n\sigma_i^2+2\sum\limits_{i<j}\rho_{ij}\sigma_{i}\sigma_{j}\right).\] 

Let us first illustrate a special situation happening when $ \Sigma $ is not full rank.
When \(0\in\Lambda(\Sigma)\), it implies that an individual's opinion is a linear combination of the opinions of the others. In this case, for our models there exists appropriate interaction matrices leading to a social power vector \(\mathbf{y}\) for which the opinions converge to the true value \(\zeta\) with probability \(1\)  (see Example~\ref{example_varianceimpossible}). 
Such $ \mathbf{y}$ is obviously the optimal solution for the wisdom of crowd problem (thereafter denoted $ \mathbf{y}^\ast$). 

\begin{lemma}
\label{lemma_zero_variance}
Let $\Sigma \succeq 0$, then \(\exists \mathbf{y}^\ast\in\Psi_1\) such that \(\mathbf{y}^\ast\Sigma\mathbf{y}^\ast=0\) \textit{iff} $\Psi_1 \cap \ker(\Sigma) \neq \emptyset$.
\end{lemma}
\begin{proof}
($\Rightarrow$) If $\Psi_1 \cap \ker(\Sigma) \neq \emptyset$, choose $\mathbf{y}^\ast \in \Psi_1 \cap \ker(\Sigma)$. Then $\mathbf{y}^{\ast\top}\Sigma\mathbf{y}^\ast = 0$.  
($\Leftarrow$) Conversely, if there exists $\mathbf{y}^\ast \in \Psi_1$ with $\mathbf{y}^{\ast\top}\Sigma\mathbf{y}^\ast = 0$, then since $\Sigma \succeq 0$, the quadratic form is zero only when $\mathbf{y}^\ast \in \ker(\Sigma)$. Hence, $\mathbf{y}^\ast \in \Psi_1 \cap \ker(\Sigma)$.\qed
\end{proof}

If $\Psi_1 \cap \ker(\Sigma) = \emptyset$, then every feasible $\mathbf{y}$ lies outside the nullspace, and since $\Sigma\succeq 0$, the minimum variance is strictly positive as shown in Example~\ref{example_varianceimpossible}. 

\begin{example}
    \label{example_varianceimpossible}
    Consider three individuals with initial opinions having mean \(\mathbb{E}[x_i(0)]=\zeta\) and a covariance matrix
 \[\Sigma_1=\begin{bmatrix}
        1 &0.5 &1.5\\
0.5& 2 &2.5\\
1.5& 2.5 &4
    \end{bmatrix}.\] 
    Since \(0\in\Lambda(\Sigma_1)\), the individuals' opinions are linearly dependent. Specifically, they obey to the constraint \(x_3(0)= x_1(0)+x_2(0)\). With the social power vector \(\mathbf{y}=[1,\, 1,\, -1]^\top \), the aggregated opinion \(\overline{x}^\ast=\mathbf{y}^\top \mathbf{x}(0)\) has a mean \(\zeta\) and zero variance. 
    Now consider a second example with
    \[
    \Sigma_2 = \begin{bmatrix}
    1 & 1 & 0 \\
    1 & 1 & 0 \\
    0 & 0 & 1
    \end{bmatrix}.
    \]
    This matrix is also singular, with \( 0 \in \Lambda(\Sigma_2) \), but  \( \mathrm{ker}(\Sigma_2)\cap \Psi_1=\emptyset \). Therefore, no social power vector \( \mathbf{y} \in \Psi_1 \) can achieve zero variance, and the minimum achievable variance is strictly positive, i.e., \(\mathrm{Var}[\overline{x}^\ast]=0.5\) with  \(\mathbf{y}^\ast=[0.25,\,0.25,\,0.5]^\top\).\qed
\end{example}

    When \( 0 \in \Lambda(\Sigma) \), if the social power vector \( \mathbf{y} \) is set to the normalized eigenvector corresponding to the zero eigenvalue, then  \( \mathbf{y}^\top  \Sigma \mathbf{y} = 0 \). This is illustrated in Example~\ref{example_varianceimpossible}, where \( [1,\, 1,\, -1]^\top  \) is the eigenvector associated with \( 0 \in \Lambda(\Sigma_1) \).
This situation is equivalent to the independent case where one of the individuals has zero variance. In fact, by applying an appropriate linear transformation that diagonalizes the covariance matrix, linear combinations of opinions become the new set of independent variables. In this transformed basis, the presence of a zero eigenvalue corresponds to one variable having zero variance, meaning it is perfectly known.
Notice that in general exploiting this special situation requires to have signed interaction matrices, as the social power vector $ \mathbf{y}$ will typically be a signed vector, as Example~\ref{example_varianceimpossible} shows.
From now on, we ignore this trivial case and consider only covariance matrices that are positive definite (p.d), i.e. \(\Sigma\succ 0\).

\subsection{Unipartite Signed Opinion Dynamics models}
Assume that we are in the signed unipartite case, i.e., $ W_s $ has the SPF property and the right dominant eigenvector $ \mathds{1}$. 
When the initial opinions are correlated, the variance of the equilibrium average \(\overline{x}^\ast\) is given by
\[
\mathrm{Var}[\overline{x}^\ast] = \mathbf{y}^\top \Sigma \mathbf{y} = \sum_{i=1}^n y_i^2 \sigma_i^2 + 2 \sum_{i<j} y_i y_j \rho_{ij} \sigma_i \sigma_j.
\]
The wisdom of crowds improves if \(\mathrm{Var}[\overline{x}^\ast] < \mathrm{Var}[\overline{x}(0)] \)
which occurs in the interior of a hyperellipsoid
\[
\Phi_3 = \left\{ n^2 \frac{\sum_{i=1}^n y_i^2 \sigma_i^2 + 2 \sum_{i<j} y_i y_j \rho_{ij} \sigma_i \sigma_j}{\sum_{i=1}^n \sigma_i^2 + 2 \sum_{i<j} \rho_{ij} \sigma_i \sigma_j} \leq 1 \right\}.
\]
Since $\mathbf{y}$ must satisfy $\mathds{1}^\top \mathbf{y} = 1$, the convex region of interest is
\(
\Gamma_6 = \Psi_1 \cap \Phi_3.
\)

The following proposition summarizes the wisdom of crowd properties in the unipartite dependent opinion case. Its proof is analogous to that of Theorem~\ref{theorem_ind_unipartite} and it is therefore omitted.
\begin{proposition}
When $ \Sigma $ is a pd covariance matrix, in the signed unipartite case all statements of Theorem~\ref{theorem_ind_unipartite} still hold provided that the concentration region \(\Gamma_2\)  is replaced by \(\Gamma_6\).

\end{proposition}

For independent opinions (diagonal $\Sigma$), the optimal social power vector is proportional to the opinion precision:
\(
y_i^\ast \propto \frac{1}{\sigma_i^2},
\)
and favors agents with lower variance.

However, when the opinions are correlated ($\Sigma$ has off-diagonal terms), the  optimal social power vector depends on the eigenstructure of $\Sigma$. If $\Sigma = U Q U^\top$, where $U=[\mathbf{u}_1,\dots,\mathbf{u}_n]$ and $Q=\mathrm{diag}([q_1,\dots,q_n])$, then
\[
\mathbf{y}^\ast = \sum_{i=1}^n \frac{(c_i/q_i)}{\sum_{j=1}^n (c_j^2/q_j)} \mathbf{u}_i,\qquad c_i = \mathbf{u}_i^\top \mathds{1}.
\]
Thus, the weights are inversely proportional to the eigenvalues \(q_i\) (variances along principal directions) and scaled by the alignment of each eigenvector with the uniform direction $\mathds{1}$. Each component can be written as
\[
y_k^\ast = \sum_{i=1}^n \frac{c_i/q_i}{\sum_{j=1}^n c_j^2/q_j} u_{ik}.
\]

When $\Sigma\succ 0$ is diagonal (independent opinion), the optimal social power vector $\mathbf{y}^\ast$ that minimizes the variance of the aggregated opinion, $\mathrm{Var}[\overline{x}^\ast]$, is always positive because $\Sigma^{-1}\mathds{1}$ has strictly positive components. 

However, when the initial opinions are strongly positively correlated, the precision matrix \(\Sigma^{-1}\)  typically has negative off-diagonal entries, which can make some row sums negative and lead to negative components in the optimal social power vector \(\mathbf{y}^\ast\). This phenomenon, illustrated in Example~\ref{example_variances}, shows that under dependent opinions, the optimal social power allocation may require negative weights to improve the wisdom of crowd.
\begin{example}
\label{example_variances}
  Consider two individuals whose initial opinions are dependent, represented by the covariance matrix \(\Sigma=\begin{bmatrix}
        2 & 10\\10 &60
    \end{bmatrix}\).  At the start of the discussion, the group variance is $\mathrm{Var}[\overline{x}(0)]=20.5$ which exceeds the minimum variance of an individual, \(\sigma_{\min}^2=2\). This indicates that there is no wisdom of the crowd initially.  To optimize the variance of the weighted average, the optimal social power vector is
 \(\mathbf{y}=\begin{bmatrix}
        1.1905  & -0.1905
    \end{bmatrix}\) which reduces the variance to
 \(\mathrm{Var}[\overline{x}^\ast]=0.4762\) thus achieving an improved wisdom of the crowd. The interaction matrix that attains this minimum variance is \(W_s=\begin{bmatrix}
        1.25 &   -0.25\\
    1.56 &   -0.56
    \end{bmatrix}\). In contrast, if we consider an interaction pattern where \(W>0\) or use eventually positive \(W_s\), then the social power vector \(\mathbf{y}\) is constrained to lie within the 2-simplex. In this case, the variance-minimizing social power vector is \(\mathbf{y}=\begin{bmatrix}
        1  & 0
    \end{bmatrix}\) resulting in a variance of
 $\mathrm{Var}[\overline{x}^\ast]=2$ which corresponds to the minimum variance of an individual, \(\sigma_{\min}^2\). \qed
\end{example}

Next, we provide a necessary and sufficient condition for the optimal social power vector \(\mathbf{y}^\ast\) to be positive.

\begin{proposition}
   Let \(\Sigma\succ 0\), then the optimal social power vector \(\mathbf{y}^\ast\) is positive \textit{iff} \(\Sigma^{-1}\mathds{1}>\mathbf{0}\). Sufficient conditions for $\mathbf{y}^\ast > \mathbf{0}$  are
    \begin{itemize}
        \item $\Sigma$ is an M-matrix (then $\mathbf{y}^\ast > \mathbf{0}$).
        \item $\Sigma \mathds{1} = \kappa \mathds{1}$ for some $\kappa > 0$ (then $\mathbf{y}^\ast = \mathds{1}/n$).
    \end{itemize}
\end{proposition}
\begin{proof}
    Since \(\Sigma\succ 0\), we have \(\Sigma^{-1}\succ 0\) and \(\mathds{1}^\top\Sigma^{-1}\mathds{1}>0\). Then, for \(\mathbf{y}^\ast\) to be positive, we require that the linear system \(\Sigma \mathbf{a} = \mathds{1}\) admits a positive solution, i.e., $\mathbf{a}=\Sigma^{-1}\mathds{1} > \mathbf{0}$. An inverse M-matrix is a nonnegative matrix, and therefore has positive row sums.
    \qed
\end{proof}

For independent opinions, the inequality \(\mathrm{Var}[\overline{x}(0)] < \sigma_{\min}^2\) holds whenever  \(\sum_{i=1}^n\sigma_{i}^2 < n^2 \sigma_{\min}^2.\) Thus, for sufficiently large $n$, the condition is satisfied, and the wisdom of crowds effect emerges. 
In contrast, for dependent opinions, this property does not generally hold. For the inequality \(\mathrm{Var}[\overline{x}(0)] < \sigma_{\min}^2\) to be satisfied, it is necessary that \( 2\sum_{i< j}\rho_{ij}\sigma_i\sigma_j < n^2\sigma_{\min}^2 - \sum_{i=1}^n \sigma_i^2,\) which may fail even as $n \to \infty$ when correlations $\rho_{ij} > 0$ are present. In fact, for positive correlations, \(\mathrm{Var}[\overline{x}(0)] \) is {typically} bigger for dependent opinions than for independent ones, as seen in Example~\ref{example_absentindependent}.

\begin{example}
    \label{example_absentindependent}
    Consider \(\Sigma=\begin{bmatrix}
        2 & 1\\1 & 1
    \end{bmatrix}\). In this case, \mbox{\(\mathrm{Var}[\overline{x}(0)]=1.25\)} which is bigger than \(\sigma_{\min}^2=1\). When \(\rho_{12}=0\), then \(\mathrm{Var}[\overline{x}(0)]=0.75\). \qed
\end{example}

Notice that for some covariance matrices it might not be possible to achieve 
\(\mathrm{Var}[\overline{x}^\ast] < \sigma_{\min}^2\) for any feasible social power vector \(\mathbf{y}\), see Example~\ref{example_nopossiblewisdom}. This occurs {for instance} when one column of the matrix \(\Sigma\) is a scalar multiple of \(\mathds{1}\); equivalently, \(\Sigma\) can be written as a rank-one update of a lower-rank matrix:
\(
\Sigma = B + \sigma_i^2 \mathds{1} \mathbf{e}_i^\top,
\)
where \(\mathbf{e}_i\) is the standard basis vector with \(1\) at the \(i\)-th position, and \(B\) satisfies \(B \mathbf{e}_i = \mathbf{0}\), i.e., \(\mathrm{rank}(B) < \mathrm{rank}(\Sigma)\). 
Moreover,  we have \(\sigma_j^2 > \sigma_i^2\) \(\forall j \neq i\), making \(\sigma_i^2\) the smallest possible variance among all agents.

\begin{corollary}
\label{corollary_rank_one}
Let $\Sigma = B + \sigma_i^2  \mathds{1} \mathbf{e}_i^\top$ where $B \mathbf{e}_i = \mathbf{0}$. Then:
\begin{itemize}
    \item 
    \(
    \mathbf{y}^\ast = \mathbf{e}_i.
    \)
    \item 
    \(
    \mathbf{y}^{\ast\top} \Sigma \mathbf{y}^{\ast} = \sigma_i^2.
    \)
    \item For all $j \neq i$, $\sigma_j^2 > \sigma_i^2$.
\end{itemize}

\end{corollary}

\begin{proof}
From the structure of $\Sigma$, we have
\[
\Sigma \mathbf{e}_i = (B + \sigma_i^2 \mathds{1} \mathbf{e}_i^\top)\mathbf{e}_i = B \mathbf{e}_i + \sigma_i^2 \mathds{1}(\mathbf{e}_i^\top \mathbf{e}_i) = 0 + \sigma_i^2 \mathds{1} = \sigma_i^2 \mathds{1},
\]
which implies \(\Sigma^{-1}\mathds{1} = \sigma_i^{-2}  \mathbf{e}_i\).
Normalizing under the constraint $\mathbf{y}^{\ast\top} \mathds{1}=1$:
\[
\mathbf{y}^\ast = \frac{\Sigma^{-1}\mathds{1}}{\mathds{1}^\top\Sigma^{-1}\mathds{1}}=\frac{\sigma_i^{-2} \mathbf{e}_i}{\mathds{1}^\top \sigma_i^{-2} \mathbf{e}_i} = \mathbf{e}_i.
\]
Thus, the minimum variance is
\[
\mathbf{y}^{\ast\top} \Sigma \mathbf{y}^\ast = \mathbf{e}_i^\top \Sigma \mathbf{e}_i = \mathbf{e}_i^\top (B + \sigma_i^2 \mathds{1} \mathbf{e}_i^\top)\mathbf{e}_i = \sigma_i^2.
\]
Finally, since in this structured case $\Sigma_{ij} = \sigma_i^2$ $\forall j$ and $\Sigma_{ij} = \rho_{ij} \sigma_i \sigma_j$ with $|\rho_{ij}| < 1$, the variance of agent $i$ is strictly smaller than that of any other agent, so $\sigma_j^2 > \sigma_i^2$ $\forall j \neq i$. \qed
\end{proof}

\begin{remark}
When $B \mathbf{e}_i = 0$,  the covariance structure effectively collapses to a single informative direction: all uncertainty in the group average is concentrated in agent \(i\). In this case, aggregation offers no benefit, i.e., the wisdom of the crowd disappears, and the only way to minimize variance is to rely entirely on agent \(i\).  Intuitively, any other agent \(j\) is either almost a copy of agent \(i\), i.e., \(\rho_{ij}\approx 1\), (adding no new information) or so noisy, i.e., \(\sigma_j^2>>\sigma_i^2\), that including it only increases variance.
\end{remark}

\begin{example}
    \label{example_nopossiblewisdom}
    Consider \(\Sigma=\begin{bmatrix}
        2 & 1\\1 & 1
    \end{bmatrix}\) with \(\sigma_{\min}^2=1\). Here, \(\Sigma=B+\mathds{1}\mathbf{e}_2^\top \) with \(B=\begin{bmatrix}
        2 & 0\\1 & 0
    \end{bmatrix}\). The social power vector that optimizes \(\mathrm{Var}[\overline{x}^\ast]=\sigma_{\min}^2=1\) is \(\mathbf{y}=\Sigma^{-1}\mathds{1}=\mathbf{e}_2\).\qed
\end{example}

\subsection{Bipartite Signed Opinion Dynamics models}
In this subsection, we assume that the interaction pattern \(W_b\) has the SSPF property. Similar to Section~\ref{section:wisdombipartitemodels}, the analysis can be done for the two cases corresponding to population-based average \(\overline{x}^\ast=\mathds{1}^\top \mathbf{x}^\ast/n\) and to bipartition-based average \(\overline{x}^\ast=\mathbf{v}^\top \mathbf{x}^\ast/n\).
Notice that for both we start from the initial population-based average (see Remark~\ref{rem:bip-dep} below)
\[
\mathrm{Var}[\overline{x}(0)] = \frac{1}{n^2}\left(\sum_{i=1}^n \sigma_i^2 + 2\sum_{ i < j } \rho_{ij}\sigma_i \sigma_j\right).
\]

\subsubsection{Population-Based Average}
For the population-based average, \(\overline{x}^\ast=\mathds{1}^\top \mathbf{x}^\ast/n\), consider a convex region \(\Gamma_7\) defined as the intersection of the hyperplane \(\Psi_2\) and the  hyperellipsoid 
$$ \Phi_4 =\left\{\frac{\sum_{i=1}^ny_i^2\sigma_i^2+2\sum_{i<j}y_iy_j\rho_{ij}\sigma_{i}\sigma_{j}}{\sum_{i=1}^n\sigma_i^2+2\sum_{i<j}\rho_{ij}\sigma_{i}\sigma_{j}}\leq  (\mathds{1}^\top \mathbf{v})^{-2} \right\}$$ $ \Gamma_7= \Psi_2 \cap \Phi_4$ and a convex region \(\Gamma_8=\Psi_3\cap\Phi_3\).

\begin{proposition}
When $ \Sigma $ is a pd covariance matrix, in the signed bipartite case and for the population-based average, all statements of Theorems~\ref{theorem_ind_bipartite_degroot} and~\ref{theorem_ind_bipartite_SFJ} still hold  provided the convex region \(\Gamma_3\) and \(\Gamma_4\) are replaced by \(\Gamma_7\) and \(\Gamma_8\).
\end{proposition}

\subsubsection{Bipartition-Based Average} For the bipartition of the groups in the average, \(\overline{x}^\ast=\mathbf{v}^\top \mathbf{x}^\ast/n\), we can  consider a convex region \(\Gamma_9\) defined as the intersection of the hyperplane \(\Psi_2\) and  the hyperellipsoid \(\Phi_3\): $ \Gamma_9= \Psi_2 \cap \Phi_3$. Here, the social power for the SFJ model is \(\mathbf{y}=P_b^\top\mathbf{v}/n\).

\begin{proposition}
When $ \Sigma $ is a pd covariance matrix, in the signed bipartite case and for the bipartition-based average, all statements of Theorem~\ref{theorem_ind_bipartite_group} still hold  provided the convex region \(\Gamma_5\) is replaced by \(\Gamma_9\).
\end{proposition}

\begin{remark}
\label{rem:bip-dep}
When computing the initial group variance in the bipartite case, we do not use
\(
{\mathbf{v}^\top \Sigma \mathbf{v}}/{n^2}.
\)
Instead, the initial variance is computed as
\(
\mathrm{Var}[\overline{x}(0)] = {\mathds{1}^\top \Sigma \mathds{1}}/{n^2},
\)
because at the beginning of the process, all agents are treated as a single population with the same mean $\zeta$. The bipartite structure (the corresponding opposite opinion values) emerges only after the discussion dynamics evolve. Therefore, using $\mathbf{v}^\top \Sigma \mathbf{v}/n^2$ at \mbox{$t=0$} would incorrectly assume that the population is already split into two groups, which is not the case initially. Bipartition-based averages become relevant only at the equilibrium stage when the two clusters have formed.
\end{remark}

\begin{table}[t]
\centering
\caption{Summary of convex regions $ \Gamma_k$}
\label{tab:Gamma6-9}
\renewcommand{\arraystretch}{1.5} 
\begin{adjustbox}{width=\columnwidth}
\begin{tabular}{c c c c}
\hline
 \textbf{Case}& \textbf{Region}  & \textbf{$\Phi = \{\mathbf{y}^\top \Sigma \mathbf{y} \le c\}$} & \textbf{$\Psi = \{\mathbf{a}^\top \mathbf{y} = t\}$} \\
\hline
 unip. + ind. / dep.  & $\Gamma_2$ / \(\Gamma_6\) &$c = \frac{\mathds{1}^\top \Sigma \mathds{1}}{n^2}$ & $\mathbf{a} = \mathds{1},\; t = 1$ \\
  bip. (pop. av.) + ind. / dep.& $\Gamma_3$ / \(\Gamma_7\)& $c = \frac{\mathds{1}^\top \Sigma \mathds{1}}{(\mathds{1}^\top \mathbf{v})^2}$ & $\mathbf{a} = \mathbf{v},\; t = 1$  \\
  bip. (bip. av.) + ind. / dep.& $\Gamma_4$ / \(\Gamma_8\) & $c = \frac{\mathds{1}^\top \Sigma \mathds{1}}{n^2}$ & $\mathbf{a} = \mathbf{v},\; t = \frac{\mathds{1}^\top \mathbf{v}}{n}$  \\
  bip. (bip. av.) + ind. / dep.& $\Gamma_5$ / \(\Gamma_9\)& $c = \frac{\mathds{1}^\top \Sigma \mathds{1}}{n^2}$ & $\mathbf{a} = \mathbf{v},\; t = 1$ \\
\hline
\end{tabular}
\end{adjustbox}
\end{table}

\section{Geometric Interpretation of Concentration Regions}
\label{section:geometricinterpretation}

The concentration regions $\Gamma_k$ for \(k=2,\dots,9\) arise as sections of ellipsoids of the form $\{\mathbf{y}:\mathbf{y}^\top \Sigma \mathbf{y}\le c\}$ by affine hyperplanes that encode different constraints on the social power vector $\mathbf{y}$ (see summary in Table~\ref{tab:Gamma6-9}).

\subsection{Geometric Interpretation of the Optimal Social Power Vector}

While Theorems~\ref{theorem_ind_unipartite}--\ref{theorem_ind_bipartite_group} (and the equivalent dependent cases as well) provide the analytical expression for the optimal social power vector $\mathbf{y}^*$, it is equally important to interpret its geometric significance.  
The following proposition establishes that $\mathbf{y}^*$ is not only the variance-minimizing point but also the geometric center (i.e., the centroid) of the corresponding ellipsoidal section.

\begin{proposition}
\label{proposition_center}
Let \(\Sigma\succ 0\) be a covariance matrix. Whenever the ellipsoidal convex region $\Gamma_k = \{\mathbf{y} \in \mathbb{R}^n : \mathbf{a}^\top \mathbf{y} = t, \ \mathbf{y}^\top \Sigma \mathbf{y} \le c\}$ is nonempty for some $\mathbf{a} \neq {\bm 0}$, $t \in \mathbb{R}$, and $c > 0$, the optimal social power vector
\[
\mathbf{y}^* = \frac{t\Sigma^{-1}\mathbf{a}}{\mathbf{a}^\top \Sigma^{-1}\mathbf{a}}
\]
is the geometric center of $\Gamma_k$. 
\end{proposition}

\begin{proof}
The quadratic form $\mathbf{y}^\top \Sigma \mathbf{y}$ defines ellipsoidal level sets under the quadratic norm $\|\mathbf{y}\|_\Sigma = \sqrt{\mathbf{y}^\top \Sigma \mathbf{y}}$. The intersection with the hyperplane $\mathbf{a}^\top \mathbf{y} = t$ is an $(n-1)$-dimensional ellipsoid.

The point $\mathbf{y}^*$ solves the strictly convex optimization problem
\(
\min_{\mathbf{y}} \ \mathbf{y}^\top \Sigma \mathbf{y} \quad \text{s.t.} \ \mathbf{a}^\top \mathbf{y} = t,
\)
yielding $\mathbf{y}^* = \dfrac{t \, \Sigma^{-1}\mathbf{a}}{\mathbf{a}^\top \Sigma^{-1}\mathbf{a}}$. For any feasible direction $\mathbf{d}$ satisfying $\mathbf{a}^\top \mathbf{d} = 0$, consider the line $\mathbf{y}(t) = \mathbf{y}^* + \tau \mathbf{d}$. Substituting into $\mathbf{y}^\top \Sigma \mathbf{y} = c$ gives
\[
(\mathbf{d}^\top \Sigma \mathbf{d}) \tau^2 + 2(\mathbf{y}^{*\top}\Sigma \mathbf{d}) \tau + (\mathbf{y}^{*\top}\Sigma \mathbf{y}^* - c) = 0.
\]
By first-order optimality, $\mathbf{y}^{*\top}\Sigma \mathbf{d} = 0$, so the roots are symmetric:
\[
\tau = \pm \sqrt{\frac{c - \mathbf{y}^{*\top}\Sigma \mathbf{y}^*}{\mathbf{d}^\top \Sigma \mathbf{d}}}.
\]
Thus, $\mathbf{y}^*$ is equidistant from the two intersection points along any feasible direction, proving it is the centroid of $\Gamma_k$.
\qed
\end{proof}

Proposition~\ref{proposition_center} implies that for any line within $\Gamma_k$ passing through $\mathbf{y}^*$ and intersecting its boundary at two points, the distances from $\mathbf{y}^*$ to these points are equal. The convex region \(\Gamma_k\) is nonempty \textit{iff} \(c>\mathbf{y}^{*\top}\Sigma \mathbf{y}^* \iff c>{t^2}/{\mathbf{a}^\top\Sigma^{-1}\mathbf{a}}\).

\subsection{Volume Ratios of Concentration Regions}
The next result quantifies the relative sizes of the concentration regions by computing the proportionality of their volumes. Specifically, it 
provides a general formula for the $(n{-}1)$-dimensional volume ratio between any two ellipsoidal sections $\Gamma_k$ and $\Gamma_\ell$, $k,l = 2, \ldots, 9$.
\begin{proposition}\label{pro:volume_ratios_general}
Let $\Sigma \succ 0$ and consider nonempty ellipsoidal convex regions
\(
\Gamma_k = \{\mathbf{y} \in \mathbb{R}^n : \mathbf{a}_k^\top \mathbf{y} = t_k,\; \mathbf{y}^\top \Sigma \mathbf{y} \le c_k\},
\)
where $\mathbf{a}_k \neq \mathbf{0}$, $t_k \in \mathbb{R}$, and $c_k > 0$. 
Then the $(n{-}1)$-dimensional volume of $\Gamma_k$ satisfies
\(
\mathrm{Vol}\big(\Gamma_k\big) \propto r_k^{{n-1}},
\)
where \(r_k=\sqrt{c_k - ({t_k^2}/{\mathbf{a}_k^\top \Sigma^{-1}\mathbf{a}_k})}\).
Consequently, for any two such regions $\Gamma_k$ and $\Gamma_\ell$, their volume ratio is
\[
\frac{\mathrm{Vol}(\Gamma_k)}{\mathrm{Vol}(\Gamma_\ell)} =
\left(
\frac{c_k - t_k^2 / (\mathbf{a}_k^\top \Sigma^{-1}\mathbf{a}_k)}{c_\ell - t_\ell^2 / (\mathbf{a}_\ell^\top \Sigma^{-1}\mathbf{a}_\ell)}
\right)^{\frac{n-1}{2}}.
\]
\end{proposition}

\begin{proof}
Consider the ellipsoid $\mathcal{E}(c)=\{\mathbf{y}:\mathbf{y}^\top \Sigma \mathbf{y}\le c\}$ and a hyperplane $\mathcal{H}(\mathbf{a},t)=\{\mathbf{y}:\mathbf{a}^\top \mathbf{y}=t\}$. Under the transformation $\mathbf{z}=\Sigma^{1/2}\mathbf{y}$, the ellipsoid becomes the Euclidean ball $\|\mathbf{z}\|^2\le c$, while the hyperplane becomes \mbox{$(\Sigma^{-1/2}\mathbf{a})^\top \mathbf{z}=t$.} The distance from the origin to this hyperplane equals
\(
d = {t}/{\sqrt{\mathbf{a}^\top \Sigma^{-1}\mathbf{a}}}.
\)
The intersection is an $(n{-}1)$-dimensional ellipsoid obtained by slicing an $n$-dimensional ball of radius $\sqrt{c}$ at distance $d$ from its center. The radius of this intersection is therefore
\(
r = \sqrt{c - d^2} = \sqrt{c - ({t^2}/{\mathbf{a}^\top \Sigma^{-1}\mathbf{a}})}.
\)
Since the intersection is $(n{-}1)$-dimensional, its volume is proportional to $r^{n-1}$. 
Taking the ratio for two such regions $\Gamma_k$ and $\Gamma_\ell$ yields the stated result.
\qed
\end{proof}

For a specific concentration region \(\mathrm{int}(\Gamma_k)\), the term \mbox{ \(   r^2= c-({t^2}/{\mathbf{a}^\top \Sigma^{-1}\mathbf{a}})\)} depends on the covariance matrix \(\Sigma\). Increasing the volume of $\Gamma_k$ enlarges this radius, which in turn increases the distance between the optimal social power vector $\mathbf{y}^\ast$ and the uniform point $\mathds{1}/n$. This relationship is evident from the derivatives of $r^2$ with respect to $\sigma_i$ and $\rho_{ij}$, which quantify how changes in $\Sigma$ affect the size of the convex region $\Gamma_k$:
\begin{align*}
    \frac{\partial r^2}{\partial \rho_{ij}}&=2\sigma_i\sigma_j((1/n^2)-Y_{ij}),\\
    \frac{\partial r^2}{\partial \sigma_{i}}&=2\sigma_i\left((1/n^2)-Y_{ii} \right)+\sum_{j\neq i}\rho_{ij}\sigma_j\left((1/n^2)-Y_{ij}\right),
\end{align*}

where \(Y=\mathbf{y}^\ast\mathbf{y}^{\ast\top}\). The effect of changing a correlation term \(\rho_{ij}\) on the volume of the concentration region depends only on the structure of \(\mathbf{y}^\ast\) (which depends on \(\rho_{ij}\) itself). In contrast, changing a variance term \(\sigma_i\) influences the volume through both \(\sigma_i\) itself and its interactions with other parameters, namely \(\rho_{ij}\) and \(\sigma_j\).

\begin{figure*}[t]
     \centering
     \begin{subfigure}[b]{0.32\linewidth}
        \centering
        \includegraphics[width=1\textwidth]{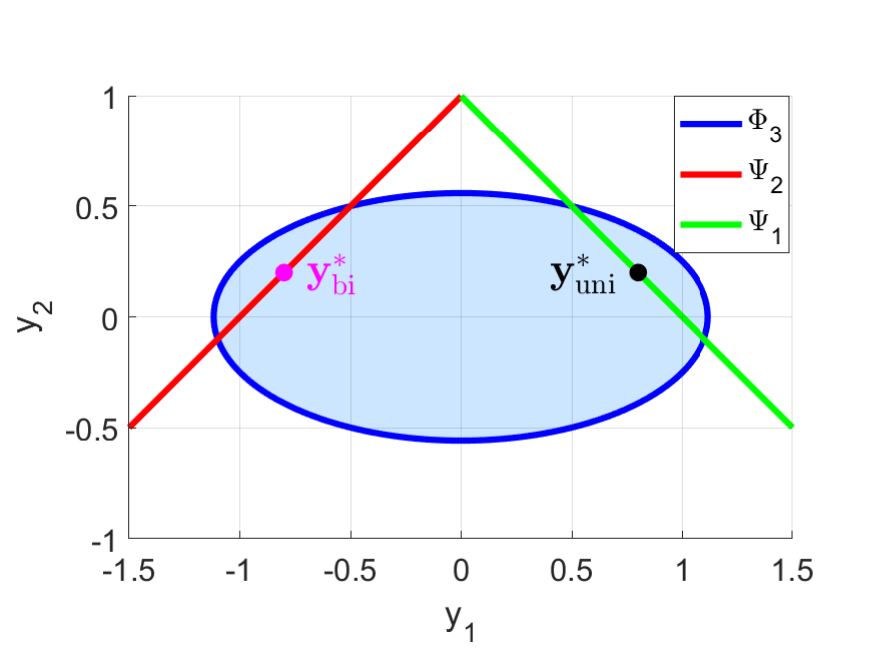}
     
      \caption{}
        \label{fig:ellipse1}
     \end{subfigure}
     \hfill
     \begin{subfigure}[b]{0.32\linewidth}
        \centering
        \includegraphics[width=1\textwidth]{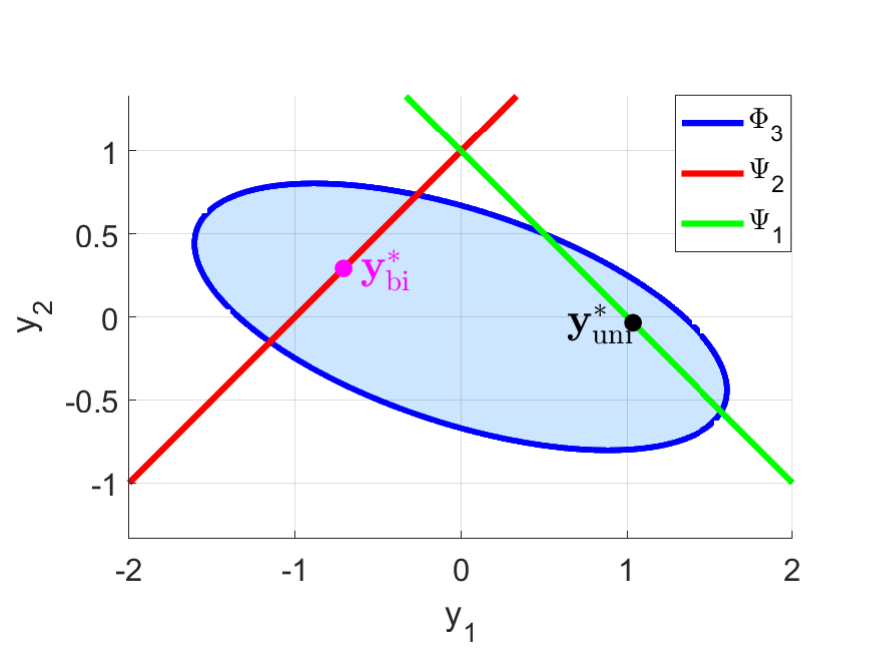}
      \caption{}
        \label{fig:ellipse2}
     \end{subfigure}8
     \hfill
     \begin{subfigure}[b]{0.32\linewidth}
        \centering
        \includegraphics[width=1\textwidth]{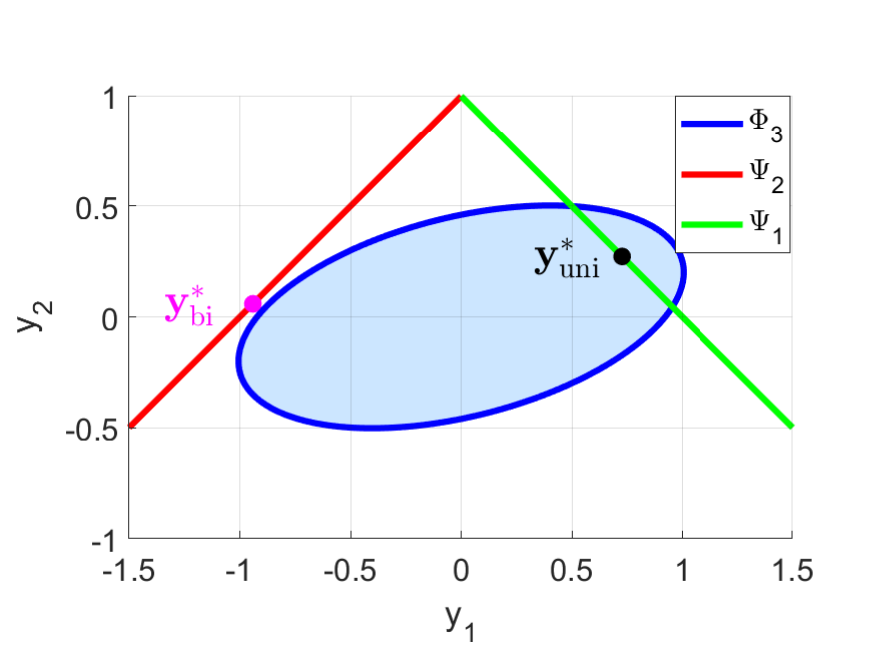}
    
       \caption{}
        \label{fig:ellipse3}
     \end{subfigure}
        \caption{
        The green line represent the hyperplane $ \Psi_1= \{ \mathds{1}^\top \mathbf{y} = 1 \}$ in $n=2$. The red line represents instead the hyperplane \(\Psi_2= \{ \mathbf{v}^\top \mathbf{y} = 1 \}\) with \(\mathbf{v}=[-1,\,1]^\top\). The convex region \(\Gamma_6\) is the intersection of the green line and the ellipse \mbox{\(\Phi_3=\{n^2\mathbf{y}^\top\Sigma\mathbf{y}\leq \mathds{1}^\top\Sigma\mathds{1}\}\),} the convex region \(\Gamma_9\) is the  intersection of the red line and the ellipse $ \Phi_3$ with variances \(\sigma_1^2=1\), \(\sigma_2^2=4\) and covariances (a) \(\rho_{12}=0\)  (b) \(\rho_{12}=0.55\) (c)  \(\rho_{12}=-0.4\). The associated optimal social powers (denoted $ \mathbf{y}^\ast_{\rm uni} $ and $ \mathbf{y}^\ast_{\rm bi} $ for the unipartite and bipartite (bipartition-based average) case), are also shown.}
        \label{fig:three_graphs}
\end{figure*}

\subsection{Properties of the concentration regions: independent case}
In the independent initial opinion case (i.e, when $ \Sigma $ is diagonal), the inequality \(c\geq t^2/(\mathbf{a}^\top\Sigma^{-1}\mathbf{a})\) holds for all concentration regions. If equality occurs, i.e., $c = {t^2}/{\mathbf{a}^\top \Sigma^{-1}\mathbf{a}}$, the region \(\Gamma_k\) degenerates to a single point, resulting in zero $(n{-}1)$-dimensional volume and making ratios in Proposition~\ref{pro:volume_ratios_general} undefined. For instance, this happens for regions \(\Gamma_2\) and \(\Gamma_5\) when \(\Sigma=\kappa I\) for some $\kappa > 0$.

The convex regions $\Gamma_3$ and $\Gamma_4$ are always nonempty, i.e., \(c_k > t_k^2 / (\mathbf{a}_k^\top \Sigma^{-1}\mathbf{a}_k)\) for all possible diagonal \(\Sigma\) where \(\mathbf{a}_3=\mathbf{a}_4=\mathbf{v}\), \(c_3=(\mathds{1}^\top\mathbf{v})^{-2}\mathds{1}^\top\Sigma\mathds{1}\), \(c_4=\mathds{1}^\top\Sigma\mathds{1}/n^2\), \(t_3=1\), and \(t_4=\mathds{1}^\top\mathbf{v}/n\), see Table~\ref{tab:Gamma6-9}. 
 Moreover, $\Gamma_2$ and $\Gamma_5$ have equal volumes whenever they are nonempty, so \mbox{$\mathrm{Vol}(\Gamma_5)/\mathrm{Vol}(\Gamma_2) = 1$} because \(c_2=c_5=\mathds{1}^\top\Sigma\mathds{1}/n^2\), \(t_2=t_5=1\), and \(\mathds{1}^\top\Sigma^{-1}\mathds{1}=\mathbf{v}^\top\Sigma^{-1}\mathbf{v}\), see again Table~\ref{tab:Gamma6-9}. In fact, the following properties hold under independence of opinion:
\begin{itemize}
    \item For any $\mathbf{y} \in \mathrm{int}(\Gamma_2)$, the vector $\mathbf{y} \odot \mathbf{v} \in \mathrm{int}(\Gamma_5)$, and the variance concentration is identical in the unipartite and bipartite cases (bipartition-based average).
    \item For any $\mathbf{y} \in \mathrm{int}(\Gamma_2)$, the vectors $\mathbf{y} \odot \mathbf{v} \in \mathrm{int}(\Gamma_3)$ or $t_4\mathbf{y} \odot \mathbf{v} \in \mathrm{int}(\Gamma_4)$, and the variance concentration in the bipartite cases (both population-based and bipartition-based average) is strictly greater than in the unipartite case.
\end{itemize}

\subsection{Properties of the concentration regions: dependent case}

Unlike the hyperellipsoids $\Phi_1$ and $\Phi_2$ in the independent case, the hyperellipsoids $\Phi_3$ and $\Phi_4$ are generally not axis-aligned when correlations are present in the initial opinions. This asymmetry arises from the cross terms $y_i y_j$ in their defining equations whenever $\rho_{ij} \neq 0$. 
This leads to fundamentally different geometric properties of the concentration regions in the dependent case compared to the independent case.

In the dependent case, the inequality \(c \geq {t^2}/{\mathbf{a}^\top \Sigma^{-1} \mathbf{a}}\) is satisfied only for the region \(\Gamma_6\), and holds with equality when \(\Sigma \mathds{1} = \kappa \mathds{1}\) for some \(\kappa > 0\). This implies that the concentration region \(\mathrm{int}(\Gamma_6)\) is nonempty if and only if \(\Sigma \mathds{1} \neq \kappa \mathds{1}\).

In contrast, for regions \(\Gamma_7\), \(\Gamma_8\), and \(\Gamma_9\), the inequality \(c < {t^2}/{\mathbf{a}^\top \Sigma^{-1} \mathbf{a}}\) may hold for some \(\Sigma\). In such cases, the corresponding concentration region is empty as shown in the following example. 

\begin{example}
\label{example_emptyregion}
  Consider two individuals with dependent initial opinions, where \(\sigma_1^2=1\), \(\sigma_2^2=4\), and the  correlation between them is \(\rho_{12}\). 
  For values of \(\rho_{12}\) in the range \mbox{ \(-0.95<\rho_{12}<-0.3\)}, the convex region \(\Gamma_9\) becomes empty, as shown in Figure~\ref{fig:ellipse3} for \(\rho_{12}=-0.4\). For example, when \(\rho_{12}=-0.4\), the initial variance of the average opinion is \(\mathrm{Var}[\overline{x}(0)]=0.85\). The optimal social power vector in the bipartite case (based on bipartition averaging) is \(\mathbf{y}^\ast=\begin{bmatrix}      0.9412  & -0.0588  \end{bmatrix}\), which leads to a final variance of \(\mathrm{Var}[\overline{x}^\ast]=0.9882\), hence no wisdom of crowds occurs.
  \qed
\end{example}

Furthermore, while in the independent case $\mathbf{y} \in \Gamma_2$ implies $\mathbf{y} \odot \mathbf{v} \in \Gamma_5$ with identical variance concentration, in the dependent case $\mathbf{y} \in \Gamma_6$ does not imply $\mathbf{y} \odot \mathbf{v} \in \Gamma_9$ since \(\mathds{1}^\top\Sigma^{-1}\mathds{1}\ne\mathbf{v}^\top\Sigma^{-1}\mathbf{v}\) in general. Even when $\mathbf{y} \odot \mathbf{v} \in \Gamma_9$, the variance concentration differs between the unipartite and bipartite cases, as illustrated in Example~\ref{example_bipartiteherd}. A similar reasoning holds for the concentration regions $\Gamma_7$ and $\Gamma_8$.

\begin{example}[Example~\ref{example_emptyregion} continued]
\label{example_bipartiteherd}
  When \(\rho_{12}=0.55\),  the convex regions \(\Gamma_6\) and \(\Gamma_9\) are shown in Figure~\ref{fig:ellipse2}. For the unipartite structure, i.e., \(\mathbf{v}=\mathds{1}\), the social power vector that optimizes the final variance is given by \mbox{\(\mathbf{y}=\begin{bmatrix}
        1.1250 &   -0.1250
    \end{bmatrix}^\top \),} resulting in a  variance of \mbox{\(\mathrm{Var}[\overline{x}^\ast]=0.9875\).} In the bipartite structure, i.e., \mbox{\(\mathbf{v}=\begin{bmatrix}
        1&-1
    \end{bmatrix}^\top \),} the social power vector (for bipartition-based average) that optimizes the variance of the weighted average is given by \(\mathbf{y}=\begin{bmatrix}
        0.5962 &   -0.4038
    \end{bmatrix}^\top \), leading to an optimized variance of \(\mathrm{Var}[\overline{x}^\ast]=0.1519\). In contrast, in the unipartite case, choosing \(\mathbf{y}=\begin{bmatrix}
        0.5962 &   0.4038
    \end{bmatrix}^\top \) results in a variance of \(\mathrm{Var}[\overline{x}^\ast]=1.2112\). \qed
\end{example}

\section{Wisdom of Crowds in Continuous Time Dynamics}
\label{section:continuous_time}

In this section,  we deal with the continuous-time (CT) versions of the signed DeGroot, SFJ, and concatenated SFJ models. Here, the signed Laplacian matrix \(L_s\) is obtained via the adjacency matrix \(A\gtreqless0\), where \(L_s=D-A\).

The analysis of wisdom of crowds for these CT models follows the same structure as in the DT case. In particular, the concentration regions and conditions for mean accuracy remain analogous to the DT. The main difference lies in the definition of the social power vector, which in the CT setting is determined by the dominant left eigenvector of the (negated) signed Laplacian matrix (or its associated matrices such as $P_{Ls}$, see below), rather than by the interaction matrix $W_s$ used in the DT case.

\subsection{Unipartite Networks}
\noindent {\bf 1): Signed DeGroot Model.} The CT signed DeGroot model is defined as:
\begin{equation}
\label{eq:CTDeGroot}
    \dot{\mathbf{x}}(t) = -L_s \mathbf{x}(t)
\end{equation}
where $L_s$ is the signed Laplacian matrix of the interaction graph. For repelling dynamics, $ L_s $ can always be obtained by translation of a matrix $ W_s$ that satisfies the SPF property, i.e., $ L_s = \phi(I- W_s)$ for some $ \phi>0$, then $ -L_s$ has properties analogous to those of $ W_s$, only referred to the spectral abscissa, rather than to the spectral radius, see \cite{razaq2025signed} for the details.

\begin{assumption}
\label{assumption_SignedDeGroot_cont}
The matrix \(I-\phi^{-1}L_s\) satisfies the SPF property for some \(\phi>0\).

\end{assumption}

\noindent {\bf 2): SFJ Model.}
The CT SFJ model incorporates stubbornness and is defined as:
\begin{equation}
\label{eq:CTSFJ}
    \dot{\mathbf{\mathbf{x}}}(t) = -((I - \Theta)L_s + \Theta) \mathbf{x}(t) + \Theta \mathbf{x}(0)
\end{equation}
where $\Theta = \mathrm{diag}(\theta_1, \dots, \theta_n)$ is the diagonal matrix of stubbornness coefficients, with $\theta_i \in [0, 1)$.

\begin{assumption}
\label{assumption_signedfjmodel_cont}
 
The matrix \(I-\phi^{-1}L_s\) satisfies the SPF property for some \(\phi>0\)
and \(\theta_i\) s.t. \(0\leq\theta_i<1\), \(\forall i\). Furthermore,  $-((I - \Theta)L_s + \Theta)$ is Hurwitz.
\end{assumption}

\noindent {\bf 3): Concatenated SFJ Model.}
In the sequence of discussions indexed by $s = 1, 2, \dots$, each discussion obeys to the CT SFJ dynamics:
\(   \dot{\mathbf{x}}(s,t) = -((I - \Theta)L_s + \Theta) \mathbf{x}(s,t) + \Theta \mathbf{x}(s,0)
\)

with the concatenation rule: \(\mathbf{x}({s+1,0}) = \lim_{t \to \infty} \mathbf{x}(s,t)\). This leads to a DT evolution in the discussion index:
\begin{equation}
\label{eq:CTconcatSFJ}
    \mathbf{x}(s+1,0) = P_{Ls} \mathbf{x}(s,0)
\end{equation}
where \( P_{Ls}=((I - \Theta)L_s + \Theta)^{-1}\Theta.\)

\begin{assumption}
    \label{assumption_signedconcatenatedfjmodel_cont}
    
    The matrix \(I-\phi^{-1}L_s\) satisfies the SPF property for some \(\phi>0\) and \(\theta_i\) s.t. \(0\leq\theta_i<1\), \(\forall i\). Furthermore,  $-((I - \Theta)L_s + \Theta)$ is Hurwitz and the matrix \(P_{Ls}\) satisfies the SPF property.
\end{assumption}
\subsubsection{Asymptotic Behavior}
All three models have been thoroughly studied in the literature, see \cite{Fontan2022,razaq2025signed}.  Here, the analysis goes beyond the case where $-L_s$ is Eventually Exponentially Positive (EEP), i.e.,  $\mathbf{z}_{Ls} > 0$, to the more general setting where $\mathbf{z}_{Ls}$ may contain negative components.

\begin{lemma}[\cite{Fontan2022,razaq2025signed}]
\label{lemma_ct_signed_models}
The following results hold for the convergence of CT signed opinion dynamics models under their respective assumptions:
\begin{enumerate}
    \item \textbf{Signed DeGroot:} If Assumption~\ref{assumption_SignedDeGroot_cont} is satisfied, then the model \eqref{eq:CTDeGroot} converges to consensus:
 \mbox{   \(
    \mathbf{x}^\ast = \mathbf{z}_{Ls}^\top \mathbf{x}(0)  \mathds{1}
    \)}
    where $\mathbf{z}_{Ls}$, with $\mathbf{z}_{Ls}^\top \mathds{1} =1$, is the dominant left eigenvector of $L_s$ satisfying $\mathbf{z}_{Ls}^\top L_s = 0$. The social power vector is $\mathbf{y} = \mathbf{z}_{Ls}$.

    \item \textbf{SFJ:} If Assumption~\ref{assumption_signedfjmodel_cont} is satisfied, then the model \eqref{eq:CTSFJ}    converges to a unique equilibrium:
    \(\mathbf{x}^\ast = P_{Ls} \mathbf{x}(0)    \)
    with $P_{Ls} \mathds{1} = \mathds{1}$ and social power vector $\mathbf{y} =  P_{Ls}^\top\mathds{1} / n$.

    \item \textbf{Concatenated SFJ:} If Assumption~\ref{assumption_signedconcatenatedfjmodel_cont} is satisfied, then the concatenated model \eqref{eq:CTconcatSFJ}    converges to consensus:
    \(
    \mathbf{x}^\ast = \lim\limits_{s \to \infty} P_{Ls}^s \mathbf{x}(0) = \mathbf{p}_{Ls}^\top \mathbf{x}(0)  \mathds{1},
    \)
    where $\mathbf{p}_{Ls}$ is the dominant left eigenvector of $P_{Ls}$, normalized such that $\mathbf{p}_{Ls}^\top \mathds{1} = 1$. The social power can be expressed as 
    \(
    \mathbf{y}= \mathbf{p}_{Ls} = (I - \Theta)^{-1} \Theta \mathbf{z}_{Ls}
    \)
    where $\mathbf{z}_{Ls}$ is the dominant left eigenvector of $L_s$.
\end{enumerate}
In all three cases, the social power vector $\mathbf{y}$ satisfies $\mathbf{y}^\top \mathds{1} = 1$, even though some components $y_i$ may be negative.
\end{lemma}

\subsection{Bipartite Networks}
In the bipartite case, the signed Laplacian $L_b$ has a simple eigenvalue at $0$ with a right eigenvector $\mathbf{v}\in\{\pm1\}^n$ (mixed signature), i.e., $L_b \mathbf{v}=\mathbf{0}$.
The Laplacian matrix here is made using the gauge transformation matrix $\Xi={\rm diag} \left(\mathbf{v}\right)$ such that  $L_s=\Xi L_b\Xi$.
Consider the associated models
\begin{enumerate}
    \item  Signed DeGroot
    \begin{equation}
\label{eq:CTDeGroot_bip}
    \dot{\mathbf{x}}(t) = -L_b\,\mathbf{x}(t).
\end{equation}
\item SFJ
\begin{equation}
\label{eq:CTSFJ_bip}
    \dot{\mathbf{x}}(t) = -\big((I-\Theta)L_b+\Theta\big)\,\mathbf{x}(t) + \Theta\,\mathbf{x}(0),
\end{equation}
\item Concatenated SFJ
\begin{equation*}
    \dot{\mathbf{x}}(s,t) = -\big((I-\Theta)L_b+\Theta\big)\,\mathbf{x}(s,t) + \Theta\,\mathbf{x}(s,0),
\end{equation*}
with concatenation $\mathbf{x}(s+1,0)=\lim\limits_{t\to\infty}\mathbf{x}(s,t)$, which yields the DT evolution
\begin{equation}
\label{eq:CTconcatSFJ_bip}
    \mathbf{x}(s+1,0)=P_{Lb}\,\mathbf{x}(s,0).
\end{equation}
where \( P_{Lb}=\big((I-\Theta)L_b+\Theta\big)^{-1}\Theta.\)

\end{enumerate}
    Consider the following associated assumptions.
    \begin{assumption}
\label{assumption_SignedDeGroot_bip_cont}
The matrix \(I-\phi^{-1}L_b\) satisfies the SSPF property for some \(\phi>0\).

\end{assumption}

 \begin{assumption}
 \label{assumption_SFJ_bip_cont}
Assumption~\ref{assumption_SignedDeGroot_bip_cont} holds. In addition, \(\theta_i\) is s.t. \(0\leq\theta_i<1\), \(\forall i\) and the matrix $-((I-\Theta)L_b+\Theta)$ is Hurwitz. 
 \end{assumption}

\begin{assumption}
\label{assumption_concatSFJ_bip_cont}
Assumption~\ref{assumption_SFJ_bip_cont} holds. In addition, the matrix \(P_{Lb} \) satisfies the SSPF property.
\end{assumption}

\subsubsection{Asymptotic Behavior}
The next lemma follows from the gauge equivalence $L_s=\Xi L_b \Xi$ and the unipartite results in Lemma~\ref{lemma_ct_signed_models} (apply $\mathbf{r}(t)=\Xi\,\mathbf{x}(t)$), see \cite{Fontan2022,razaq2025signed}.

\begin{lemma}[\cite{Fontan2022,razaq2025signed}]
\label{lemma_ct_bip}
The following results hold for the CT bipartite models:
\begin{enumerate}
    \item \textbf{Signed DeGroot:} Under Assumption~\ref{assumption_SignedDeGroot_bip_cont}, the system \eqref{eq:CTDeGroot_bip} converges to a bipartite consensus
    \(
       \mathbf{x}^\ast = \mathbf{z}_{Lb}^\top \mathbf{x}(0)\mathbf{v},
    \)
    where $\mathbf{z}_{Lb}$ is the dominant left eigenvector of $L_b$ normalized such that $\mathbf{z}_{Lb}^\top \mathbf{v}=1$. The social power vector is $\mathbf{y}=\mathbf{z}_{Lb}$.

    \item \textbf{SFJ:}  Under Assumption~\ref{assumption_SFJ_bip_cont}, the system \eqref{eq:CTSFJ_bip} converges to
    \(
       \mathbf{x}^\ast = P_{Lb}\mathbf{x}(0),
    \)
    where $P_{Lb}\mathbf{v}=\mathbf{v}$. The social power vector is $\mathbf{y}= P_{Lb}^\top\mathds{1}/n$, satisfying $\mathbf{y}^\top\mathbf{v}=\mathds{1}^\top \mathbf{v}/n$.

    \item \textbf{Concatenated SFJ:}  Under Assumption~\ref{assumption_concatSFJ_bip_cont}, the system \eqref{eq:CTconcatSFJ_bip} converges to
\mbox{\(       \mathbf{x}^\ast=\lim\limits_{s\to\infty}P_{Lb}^s\mathbf{x}(0)=\mathbf{p}_{Lb}^\top \mathbf{x}(0)\mathbf{v},    \)}
    where $\mathbf{p}_{Lb}$ is the dominant left eigenvector of $P_{Lb}$ normalized by $\mathbf{p}_{Lb}^\top \mathbf{v}=1$. Moreover,
   \mbox{ \(
       \mathbf{p}_{Lb}=(I-\Theta)^{-1}\Theta\,\mathbf{z}_{Lb}.
    \)}
\end{enumerate}
In all three cases, some components of $\mathbf{y}$ may be negative.
\end{lemma}

\subsection{Analysis of Wisdom of Crowd}
In both the unipartite and the bipartite case, the analysis of the wisdom of crowds in the CT setting follows the same pattern as in the DT case, as summarized in the following proposition.
\begin{proposition}
For the CT models \eqref{eq:CTDeGroot},~\eqref{eq:CTSFJ} and \eqref{eq:CTconcatSFJ}, the statements of Theorem~\ref{theorem_ind_unipartite} for unipartite case and Theorems~\ref{theorem_ind_bipartite_degroot},~\ref{theorem_ind_bipartite_SFJ} and~\ref{theorem_ind_bipartite_group} for the bipartite case still hold, provided that the social powers of the models are replaced by those given in Lemma~\ref{lemma_ct_signed_models} and~\ref{lemma_ct_bip}.
\end{proposition}

\section{Conclusion}
\label{section:conclusion}
In this paper, we analyze the conditions under which the wisdom of crowds gets improved or undermined by a linear opinion dynamics model with antagonistic interactions. We show that the region of improvement is always larger than in the unsigned case because improvement can occur also for negative social powers. 
Moreover, we show that in a bipartite setting the crowd typically converges to a false truth, around which the variance tends to concentrate faster than in the unipartite case. The resulting signed social power provides also a natural framework to investigate the more general case of correlated initial opinions.

\bibliographystyle{IEEEtran}
\bibliography{References}

\end{document}